\documentclass{article}
\usepackage{amsmath, amsfonts, amsthm, amssymb, amscd,verbatim}
\usepackage{hyperref}
\usepackage{graphicx,color}
\usepackage{mathrsfs} 
\newtheorem{theorem}{Theorem}[section]
\newtheorem{proposition}[theorem]{Proposition} 
\newtheorem{lemma}[theorem]{Lemma}
\newtheorem{corollary}[theorem]{Corollary} 
\newtheorem{definition}[theorem]{Definition} 
\newtheorem{remark}[theorem]{Remark} 
  
\newcommand \la \langle
\newcommand \ra \rangle

\newcommand \del	\partial

\newcommand \eps 	\epsilon 

\newcommand \be 		{\begin{equation}}
\newcommand \ee 		{\end{equation}}

\newcommand{\ras}{r_{ASch}}

\newcommand{\frf}{\mathfrak{F}}

\let\oldmarginpar\marginpar
\renewcommand\marginpar[1]{\-\oldmarginpar[\raggedleft\footnotesize #1]%
{\raggedright\footnotesize #1}}
\author{Gustav Holzegel\footnote{Princeton University, Department of Mathematics, Fine Hall, Washington Road, Princeton, NJ 08544, United States.}
\, and Jacques Smulevici\footnote{
Max-Planck-Institut f\"ur Gravitationsphysik, 
Albert-Einstein Institut, Am M\"uhlenberg 1, 14476 Golm, Germany. }
}
\title{Stability of Schwarzschild-AdS for the spherically symmetric Einstein-Klein-Gordon system}

\begin{document}
\maketitle
\begin{abstract}
In this paper, we study the global behavior of solutions to the spherically symmetric coupled Einstein-Klein-Gordon (EKG) system in the presence of a negative cosmological constant. We prove that the Schwarzschild-AdS spacetimes (the trivial black hole solutions of the EKG system for which $\phi=0$ identically) are asymptotically stable: Small perturbations of Schwarzschild-AdS initial data again lead to regular black holes, with the metric on the black hole exterior approaching a Schwarz\-schild-AdS spacetime.
The main difficulties in the proof arise from the lack of monotonicity for the Hawking mass and the asymptotically AdS boundary conditions, which render even (part of) the orbital stability intricate. These issues are resolved in a bootstrap argument on the black hole exterior, with the redshift effect and weighted Hardy inequalities playing the fundamental role in the analysis. Both integrated decay and pointwise decay estimates are obtained.
\end{abstract}

\tableofcontents

\section{Introduction}
The subject of black hole stability has undergone rapid development in the past ten years. There is now an extensive literature addressing the behavior of linear wave equations on black hole backgrounds \cite{DafRod2, Sterbenz, Toha1, Mihalisnotes,DafRodKerr, Toha2, AndBlue, DafRodsmalla}, the current state of the art being a decay result for $\psi$ satisfying $\Box_g \psi =0$ with $g$ being the metric of a subextremal ($|a|<M$) Kerr spacetime \cite{DafRodlargea}. In addition, there has been progress in developing techniques to address non-linear problems on fixed backgrounds \cite{Luk1,Toha3} and, most recently, preliminary attempts to bridge the gap between these linear and the prospective full non-linear stability problem \cite{Holzegelults}.  

\subsection{Scalar waves and black hole stability problems}
While the problem of Kerr stability will, presumably, require the study of the Bianchi equations as in \cite{ChristKlei}, rather than the scalar wave equation, there is nonetheless a coupled non-linear gravitational system, whose metric evolution is governed entirely\footnote{By this we mean that if $\psi=0$, the spacetime is stationary. The metric $g$ depends nonetheless non-linearly on $\psi$ via the Einstein equations.} by a scalar field $\psi$ satisfying 
\begin{eqnarray} \label{eq:fwe}
\Box_g \psi=0.
\end{eqnarray}
 This is the well-known spherically-symmetric coupled Einstein-scalar field system
\begin{align} \label{ES}
R_{\mu \nu}-\frac{1}{2}g_{\mu \nu}R + \Lambda g_{\mu \nu} =8 \pi \mathbb{T}_{\mu \nu} = 8\pi \left[ \partial_\mu \psi \partial_\nu \psi-\frac{1}{2}g_{\mu \nu} (\partial \psi )^2 \right] 
\end{align}
with $\Lambda$ being the cosmological constant and $g_{\mu \nu}$ a spherically symmetric metric.
The study of this system was initiated in the 1960ies. Over the years, a complete and satisfactory picture of the dynamics has emerged for the asymptotically flat case ($\Lambda=0$). In a sequence of papers \cite{Christodoulou, Christodoulou2, Christodoulou3, Christodoulou4}, Christodoulou proved that generic initial data either evolve into regular black holes or the solution is geodesically complete, with the Bondi mass approaching zero along null-infinity. Moreover, it has been shown in \cite{DafRod} that the metric of black hole solutions must eventually asymptote to a Schwarzschild metric on the domain of outer communication. 

The paper\footnote{Note that in \cite{DafRod}, the system studied is actually the spherically symmetric Einstein-Maxwell-scalar field equations, which reduce to \eqref{eq:fwe}-(\ref{ES}) if the Maxwell field vanishes.}\cite{DafRod} also provides polynomial decay rates for the radiation, commonly known as ``Price's law" in the literature. While \cite{DafRod} is a large data result, it exploits extensively that under spherical symmetry \eqref{eq:fwe}-\eqref{ES} reduce to a system of 1+1 dimensional PDEs by invoking the special analytical  tools available in this setting. On the other hand, it was shown in \cite{DafHol, HolBi9} that the vectorfield techniques developed for the linear wave equation on a \emph{fixed} Schwarzschild background are sufficiently robust to understand the dynamics of the coupled system \eqref{eq:fwe}-(\ref{ES}) in a neighborhood of Schwarzschild, that is to say to prove asymptotic stability of Schwarzschild within this model.\footnote{We remark that the paper \cite{HolBi9} studies a five-dimensional version of the spherically-symmetric Einstein scalar field system, more precisely,  a class of vacuum, $SU\left(2\right)$-symmetric spacetimes, also known as biaxial Bianchi IX. The method, however, easily specializes to the spherically-symmetric coupled scalar field system in four dimensions.} 
This avoids the use of techniques which are special to 1+1 dimensional PDEs and
connects, in a satisfactory manner, the numerous works on the linear wave equation with a non-linear model of gravitational collapse, illustrating at the same time how appropriate the vectorfield method is for these types of non-linear applications.

\subsection[Linear wave equations in asymptotically dS and AdS spacetimes]{Linear wave equations in asymptotically de-Sitter and Anti-de-Sitter spacetimes}
It is natural to ask how these results change when $\Lambda \neq 0$ in (\ref{ES}). For a positive cosmological constant, $\Lambda>0$, the system (\ref{ES}) has been studied (without symmetry assumptions and for more general matter models) to prove stability of the trivial solution, de Sitter space \cite{Friedrich, Ringstrom}. In the black hole context, there has also been work on the linear wave equation $\Box_g \psi = 0$ for $g$ a fixed Schwarzschild-de Sitter metric \cite{DafRoddS, Haefner, Melrose} and more recently, Kerr-de Sitter metric \cite{Dyatlov1, vd:maahkds, Dyatlov2}.

For a negative cosmological constant, $\Lambda = -\frac{3}{l^2}<0$, there are only few results available, the linear problem having recently been addressed in \cite{HolzegelAdS}. In the latter paper, the massive wave equation,
\begin{align} \label{wap}
\Box_g \psi - \frac{2a}{l^2} \psi = 0 \, ,
\end{align}
with the mass of the Klein-Gordon field $a$ satisfying the Breitenlohner-Freedman (BF) bound $-a < \frac{9}{8}$, is studied for a class of stationary spacetimes $\left(\mathcal{M},g\right)$, which are sufficiently close to a slowly rotating Kerr-AdS spacetime. A boundedness result is then proven for a certain class of solutions to (\ref{wap}). The existence and uniqueness (after imposing suitable boundary conditions) of this class of solutions to (\ref{wap}) on any asymptotically AdS spacetime was only assumed in \cite{HolzegelAdS}, with a proof now available in \cite{Holzegelwp}. (See also \cite{Vasy2}, which likewise proves well-posedness of  \eqref{wap} with Dirichlet conditions for asymptotically AdS spacetimes admitting a conformal compactification.). Note that even this local well-posedness statement is non-trivial in view of the non globally-hyperbolic nature of these spacetimes. We refer to the introduction of \cite{HolzegelAdS}, as well as the original work \cite{Breitenlohner} for more motivation and an explanation of the BF-bound.

\subsection{The spherically symmetric Einstein-Klein-Gordon system}
In the present paper, we study the corresponding non-linear coupled gravitational system, the so-called Einstein-Klein-Gordon (EKG) system. As in the asymptotically flat case, the metric evolution is governed by (\ref{wap}) in the sense that if $\psi=0$, then the only solutions are the Schwarzschild-Anti-de-Sitter spacetimes or Anti-de-Sitter, by a simple generalization of Birkhoff's theorem. 

Hence, we are interested in triples $(\mathcal{M},g,\phi)$ such that $(\mathcal{M},g)$ is a spherically-symmetric spacetime satisfying 
\begin{align} \label{EKG}
R_{\mu \nu}-\frac{1}{2}g_{\mu \nu}R+\Lambda g_{\mu \nu}=8 \pi \mathbb{T}_{\mu \nu}, 
\end{align}
\begin{align} \label{emEKG}
\mathbb{T}_{\mu \nu}=\partial_\mu \psi \partial_\nu \psi-\frac{1}{2}g_{\mu \nu} (\partial \psi )^2-\frac{a}{l^2} \psi^2 g_{\mu \nu},
\end{align}
and such that $\psi$ satisfies the Klein-Gordon equation \eqref{wap} associated to $(\mathcal{M},g)$ with mass $a \in \mathbb{R}$ and where we recall that $\Lambda=-\frac{3}{l^2}$, with $l \in \mathbb{R}$. As for the linear case, we shall need a bound on $a$. Here, we will assume that $a$ is below the conformal case (the latter being included):
\begin{equation} \label{confbound}
a \ge -1 \, .
\end{equation}
In fact, several results of this paper will hold under the BF-bound $a > -\frac{9}{8}$ since the only argument which really exploits $a\ge-1$ is contained in the proof of the integrated decay estimate (Proposition \ref{intdec3}).

The study of the EKG system with $\Lambda < 0$ was initiated in \cite{gs:lwp}, where we prove that this system is well-posed under appropriate regularity and boundary conditions, with the time of existence of solutions depending only on an invariant $H^2$-type norm for the Klein-Gordon-field. As applications, we formulated two extension principles and obtained the existence and uniqueness of a maximal development. These results form the basis for the analysis conducted in this paper. 

\subsection{The Stability of Schwarzschild-Anti-de-Sitter spacetimes}
The main result of this paper may be summarized as follows:
\begin{theorem}
The Schwarzschild-Anti de Sitter spacetime is both orbitally and asymptotically stable within EKG in spherical symmetry provided the mass satisfies (\ref{confbound}). 
\end{theorem} \label{th:mtheorem}
In other words, small perturbations of Schwarzschild-AdS spacetimes again evolve into black hole spacetimes with a regular event horizon and a complete null-infinity (orbital stability).  Moreover, the scalar field decays towards the future (asymptotic stability). A more precise version of Theorem \ref{th:mtheorem}, which includes in particular the specific decay rates is given in Section \ref{se:tmt} (Theorem \ref{maintheorem}).

We spend the remainder of this introduction to discuss the techniques entering the proof of this theorem. 
As in the asymptotically-flat case, the system \eqref{wap}-(\ref{EKG})-\eqref{emEKG} comes with a conservation law, which is manifest in the properties of a generalized Hawking mass (renormalized by a cosmological term). Contrary to the former case case, however, this Hawking mass does not enjoy any monotonicity properties, in view of the possibly negative zeroth order mass term in (\ref{emEKG}). This can also be understood from the vectorfield point of view: There is a certain geometric vectorfield, $T$, which is not Killing but nevertheless gives rise to a conservation law via the energy momentum tensor of $\psi$. This is the well-known Kodama vectorfield in spherical symmetry \cite{Kodama}. From this perspective, the failure of monotonicity is simply the energy momentum tensor (\ref{emEKG}) not satisfying the dominant energy condition. In any case, this behavior implies that the extension principle developed for the asymptotically-flat case (see \cite{DafRen,Mihali1, Mihali2}) is not available, since it relied on the monotonicity properties of the Hawking mass. While the general structure of the Penrose diagram can in fact still be inferred from a generalized extension principle (originally developed in \cite{CamJon} and which was adapted to our problem in \cite{gs:lwp}),  the lack of monotonicity turns proving the completeness of null-infinity (which is part of the \emph{orbital} stability statement of Schwarzschild-AdS within \eqref{wap}-(\ref{EKG})-\eqref{emEKG} into a difficult problem, while for (\ref{ES}) it followed (for any data containing a trapped surface) from the extension principle and the monotonicity alone.

Studying the massive wave equation on Schwarzschild in \cite{HolzegelAdS}, it was observed that while the energy density is not pointwise positive definite, one can still prove a global doubly-weighted Hardy inequality on spacelike slices, allowing one to absorb the zeroth order term by a derivative term and hence to establish positivity in an integrated sense. It turns out that for the non-linear system \eqref{wap}-(\ref{EKG})-\eqref{emEKG}, this integrated monotonicity survives in a region away from the future event horizon (whose location is bootstrapped in the non-linear problem under consideration). Close to the horizon, on the other hand, we have the redshift available, which will allow us to control the wrong-signed zeroth order term. We will see two manifestations of the redshift here: One in the framework of vectorfields (and hence $L^2$-type estimates), the other in the context of pointwise estimates along characteristics (originally developed in \cite{DafRod}), the latter being a typical feature of 1+1 dimensional systems.

It is instructive to compare this situation with the case of the linear wave equation on slowly rotating asymptotically-flat Kerr spacetimes. For these spacetimes, there is a conserved energy  associated with the Killing field $\partial_t$, which is negative in a region close to the horizon, as the latter vectorfield becomes spacelike there. This is the well-known ergo-sphere and the phenomenon of superradiance that it triggers. One of the main insights of \cite{DafRodKerr} was that this problem can be resolved by exploiting the redshift effect as a stability mechanism near the horizon. It is very much in this fashion that we are able to remedy the problems near the horizon for \eqref{wap}-(\ref{EKG})-\eqref{emEKG}. Our setting is easier in that we are dealing with a highly symmetric problem (in particular, there is no trapping!), but at the same time it contains new difficulties since we are addressing a fully coupled problem and since here the pointwise non-positivity is actually a global feature.

Following this strategy, which is unfolded in a bootstrap argument on the location of the event horizon, we will prove that the scalar field $\psi$ remains small outside the future event horizon, provided the initial data is chosen sufficiently small. In a second step, we prove an integrated decay estimate, which implies that $\psi$ has to decay towards the future. 
We remark that the existence of such an integrated decay estimate is already suggested from the global Hardy inequalities proven in \cite{HolzegelAdS} in the linear setting, see Appendix \ref{noha}. 
 
 This almost completes the proof of the theorem except for an issue which has to do with the radial decay towards null-infinity. We recall from \cite{gs:lwp} that the well-posedness statement is formulated in terms of a weighted $H^2$-norm for $\psi$. To establish global existence in this functional space, we need to also commute with the vectorfield $T$ mentioned above. This introduces a fair amount of error-terms (as $T$ is not Killing) which are, however, easily controlled from previous bounds and by adding an $\epsilon$ of the integrated decay estimate that we can prove simultaneously for $T\psi$. It then follows that $T\psi$ satisfies the same boundedness and decay estimates as $\psi$ does, which in particular establishes improved decay for $T\psi$. From this, the radial decay for \emph{all} first derivatives can be improved, depending on how close $a$ is to the Breitenlohner-Freedman bound (the closer, the smaller the improvement). This last step uses a version of the redshift effect which is present for asymptotically AdS spacetimes near null-infinity.\footnote{This improvement arising after commutation has been exploited in \cite{Holzegelwp} in the context of weighted elliptic estimates on spatial slices and in \cite{gs:lwp} in the context of pointwise estimates.}  Collecting the improvements one recovers global uniform boundedness in the $H^2$-spaces of \cite{gs:lwp}. 
 
Finally, one proves that the integrated decay estimate implies exponential decay of the energy. This relies on the characteristic $r$-weights in the energy of asymptotically AdS spacetimes: Unlike in the asymptotically-flat context, the integrated decay estimate here controls the energy integrated in time without any loss of $r$-weights. From the statement that the energy integrated in time is controlled by the energy itself, exponential decay follows.\footnote{Note that the method of proof used in this paper to derive the decay statements of Section \ref{se:tmt} may naturally be applied to the linear case, cf.~Corollary \ref{lincol}.}
  
In summary, the paper settles (excluding questions about black hole interiors) the issue of global dynamics for (EKG) for the mass range \eqref{confbound} and with Dirichlet boundary conditions for $\psi$ near the Schwarzschild-AdS solution.

\subsection{Hairy black holes and motivations from high-energy physics}
We conclude this introduction by providing some background information about the system \eqref{wap}-(\ref{EKG})-\eqref{emEKG}. Indeed, another important motivation for the study of this system derives from high energy physics, more precisely, the AdS-CFT correspondence and its potential applications to condensed matter physics  \cite{Winstanley,Hartnoll, Gauntlett}. In this field, systems such as \eqref{wap}-(\ref{EKG})-\eqref{emEKG} (coupled often also to electromagnetism or complex scalar fields) are considered as models describing phase transitions in superconductors. From the gravitational point of view, such phase transitions correspond to non-trivial (i.e.~with non-identically vanishing scalar field) stationary black hole spacetimes, also known as black holes with ``scalar hair". Such solutions are possible in principle because the ``no-hair" theorems valid in the asymptotically flat case do not (typically) generalize to the case of a negative cosmological constant. The main result of this paper excludes the existence of such hairy black holes, in a neighboorhood of the the Schwarzschild-AdS solution, within the class of boundary conditions on $\psi$ considered. 

As suggested from asymptotic expansions of (\ref{wap}), there is an important alternative class of boundary conditions, Neumann-conditions for $\psi$.\footnote{at least in the range $\frac{5}{8} \leq -a < \frac{9}{8}$. For $-a < \frac{5}{8}$, there is only one solution and no freedom to impose boundary conditions.} For this class, one could attempt to carry out a similar program as in \cite{HolzegelAdS, gs:lwp}: Prove a well-posedness statement for (\ref{wap}) and a thereafter for the non-linear \eqref{wap}-(\ref{EKG})-\eqref{emEKG} now imposing Neumann conditions at the boundary. The global dynamics under these circumstances may be more complicated, even in a neighborhood of Schwarzschild. In particular, the aforementioned hairy black hole solutions may enter the picture. We postpone the analysis of this system to future work.

\subsection{Outline}
The outline of this paper is as follows. Section \ref{se:pre} contains all the necessary background material needed to carry out the analysis of this paper. In particular, some properties of the Schwarzschild-AdS solution are presented (Section \ref{se:sads}), a functional framework adapted to our problem is introduced (Section \ref{se:norm}) and the existence of a maximal development is recalled (Section \ref{se:tmd}). After these preliminaries, we present a detailed version of our main results in Section \ref{se:tmt}. In Section \ref{mp1}, we derive integrated decay estimates via vectorfield methods. In Section \ref{se:vfid}, we obtain higher-order estimates as well as improved decay estimates. Finally, in the last section of the paper, we conclude the proof of our main result by deriving exponential decay estimates. 
In Appendix \ref{noha}, we present an independent result, establishing the non-existence of stationary solutions for the linear wave equation on Schwarzschild-AdS backgrounds satisfying the boundary conditions of \cite{HolzegelAdS, Holzegelwp}.

\section{Preliminaries} \label{se:pre}
\subsection[The Einstein-Klein-Gordon equations in null coordinates]{The spherically-symmetric Einstein-Klein-Gordon\\ system in double null coordinates}
We start by recalling a standard result concerning the warped product structure of the metric for spherically symmetric solutions and the form of the equations in double null coordinates:

\begin{lemma} \label{lem:pre}
Let $(\mathcal{M},g,\phi)$, with $(\mathcal{M},g)$ a $C^2$ Lorentzian manifold, dim $\mathcal{M}=4$ and $\phi$ a $C^2(\mathcal{M})$ function, be a solution to the system \eqref{wap}-\eqref{EKG}-\eqref{emEKG}. Assume that  $(\mathcal{M},g,\phi)$ is invariant under an effective action of $SO(3)$ with principal orbit type a $2$-sphere. Denote by $r$ the area-radius of the spheres of symmetry. Then, locally around any point of $\mathcal{M}$, there exist double null coordinates $u,v$ such that the metric takes the form:
\begin{eqnarray} \label{eq:metric}
g=-\Omega^2 du dv + r^2 d\sigma_{S^2},
\end{eqnarray}
where $\Omega$ and $r$ may be identified with $C^2$ functions depending only on $(u,v)$ and where $d\sigma_{S^2}$ denotes the standard metric on $S^2$.
Let $\mathcal{Q}=\mathcal{M}/SO(3)$ denote the quotient of the spacetime by the orbits of symmetry. 
Then, the Einstein-Klein-Gordon equations\footnote{By a small abuse of notation, we denote functions on $\mathcal{M}$ and their projections to $\mathcal{Q}$ by the same symbols.} reduce to:
\begin{eqnarray}
\partial_u \left( \frac{r_u}{\Omega^2} \right) &=&-4\pi r \frac{(\partial_u \phi)^2}{\Omega^2}, \label{cons1} \\
\partial_v \left( \frac{r_v}{\Omega^2} \right) &=&-4\pi r \frac{(\partial_v \phi)^2}{\Omega^2}, \label{cons2} \\
r_{uv}&=&-\frac{\Omega^2}{4r}-\frac{r_u r_v}{r}+4\pi r (\frac{a \Omega^2 \phi^2}{2l^2})+\frac{1}{4}r \Omega^2 \Lambda, \label{eq:ruv}\\
 \left (\log \Omega \right)_{uv} &=& \frac{\Omega^2}{4r^2}+\frac{r_u r_v }{r^2}-4 \pi \partial_u \phi \partial_v \phi, \\
\partial_u \partial_v \phi  &=& -\frac{r_u}{r} \phi_v - \frac{r_v}{r} \phi_u - \frac{\Omega^2 a}{2l^2} \phi. \label{laste}
\end{eqnarray}
\end{lemma}
Note that (\ref{cons1}) and (\ref{cons2}) are the Raychaudhuri equations governing the evolution of area of the spheres of symmetry. Note also that the last equation is simply the the wave operator associated with $g$ acting on spherically symmetric scalar fields, which may be written shorthand as
\begin{align}
0 = \Box_g \phi - \frac{2a\phi}{l^2} = -\frac{4}{\Omega^2} \left( \partial_u \partial_v \phi +\frac{r_u}{r} \phi_v + \frac{r_v}{r} \phi_u \right) - \frac{2a\phi}{l^2}.
\end{align}
We shall use the following first order notation:
\begin{align} \label{eq:fon}
r_u= \nu \ \ ; \ \ r_v= \lambda  \ \ ; \ \ 
r \phi_u = \zeta  \ \ ; \ \ r \phi_v=  \theta \ \ ; \ \ \kappa = -\frac{\Omega^2}{4r_u} \ \ ; \ \  \gamma = \frac{\Omega^2}{4r_v} \, .
\end{align}
We can then rewrite the Raychaudhuri equations as
\begin{align} \label{kappaR}
\partial_u \kappa = - \frac{\Omega^2}{\nu^2} r \pi (\partial_u \phi )^2 
= -\frac{16}{\Omega^2} \kappa^2 r \pi  (\partial_u \phi )^2 < 0  \ \ \ ; \ \ \ 
\partial_u \log \kappa = \frac{4 \pi r}{\nu} (\partial_u \phi )^2 
\end{align}
and
\begin{align} \label{gammaR}
\partial_v \log \gamma= \frac{4 \pi r}{\lambda} (\partial_v \phi )^2 .
\end{align}
We define the (renormalized) Hawking mass,
\begin{eqnarray}
\varpi = \frac{r}{2}\left(1+\frac{4r_u r_v}{ \Omega^2} \right) - \frac{\Lambda}{6} r^3 = \frac{r}{2}\left(1+\frac{4r_u r_v}{ \Omega^2} \right) + \frac{r^3}{2l^2} \, , 
\end{eqnarray}
which is seen to satisfy
\begin{eqnarray} 
\partial_u \varpi = -8 \pi r^2 \frac{r_v}{\Omega^2} (\partial_u \phi )^2+\frac{4\pi r^2 a }{l^2} r_u \phi^2 \, , \label{ee:varpiu}\\
\partial_v \varpi = -8 \pi r^2 \frac{r_u}{\Omega^2} (\partial_v \phi )^2+\frac{4\pi r^2 a }{l^2} r_v \phi^2 \, . \label{ee:varpiv}
\end{eqnarray}
We finally collect some further identities, which will be useful to refer to in later computations.
The volume element is $\sqrt{-g} =\frac{\Omega^2 r^2 }{2}$ and hence
\begin{align}
\sqrt{-g} g^{uv}= \frac{\Omega^2 r^2}{2} \left(\frac{-2}{\Omega^2}\right)=- r^2.
\end{align}
The following identities hold for the Christoffel symbols: $\Gamma^a_{\phantom{a}uv}=0$ for $a=\{u,v,\theta,\phi\}$ and
\begin{align}
\Gamma^u_{\phantom{u}uu} = g^{uv} \left(g_{uv}\right)_u = \frac{2\Omega_u}{\Omega}   \textrm{ \ \ \ \ ,   \ \ \ \ } \Gamma^v_{\phantom{v}vv} = g^{uv} \left(g_{uv}\right)_v = \frac{2\Omega_v}{\Omega}  \, ,
\end{align}
\begin{align}
g^{ab} \Gamma^u_{\phantom{u}ab} = \frac{4r_v}{r \Omega^2} = \frac{1}{r \gamma} \textrm{ \ \ \ \ \ \ \ \ , \ \ \ \ \ \ } g^{ab} \Gamma^v_{\phantom{v}ab} = \frac{4r_u}{r \Omega^2} = - \frac{1}{r \kappa} \, .
\end{align}
The square of the gradient of $\phi$ is
\begin{align}
g(\nabla \phi, \nabla \phi ) = \frac{- 4}{\Omega^2} \phi_u \phi_v .
\end{align}
Finally, the wave equation for $r$ may be rewritten using the Hawking mass $\varpi$:
\begin{equation}
r_{uv}=-\frac{\Omega^2 \varpi}{2r^2}-\frac{\Omega^2 r}{2 l^2}+\frac{2\pi r a \Omega^2 \phi^2}{l^2} \, .
\end{equation}
\subsection{Schwarzschild-AdS} \label{se:sads}
We recall that, by definition, a Schwarzschild-AdS spacetime is any maximally analytically extended spherically-symmetric solution of the vacuum Einstein equations $R_{\mu \nu}=\Lambda g_{\mu \nu}$, with $\Lambda < 0$, excluding the Anti-de-Sitter spacetime. For a given $\Lambda$, the set of all such solutions forms a one parameter family, the parameter typically called the mass and denoted by $M$.

Consider a maximally extended Schwarzschild-AdS spacetime of mass $M$ with cosmological constant $\Lambda = \frac{-3}{l^2}$. Its Penrose diagram is depicted below.\footnote{We refer the reader to appendix C of \cite{DafRod} for a formal introduction to Penrose diagrams.}
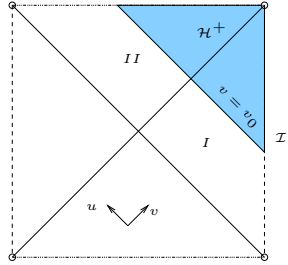
\[
\begin{picture}(0,0)%
\includegraphics{penAdS0.pstex}%
\end{picture}%
\setlength{\unitlength}{1737sp}%
\begingroup\makeatletter\ifx\SetFigFont\undefined%
\gdef\SetFigFont#1#2#3#4#5{%
  \reset@font\fontsize{#1}{#2pt}%
  \fontfamily{#3}\fontseries{#4}\fontshape{#5}%
  \selectfont}%
\fi\endgroup%
\begin{picture}(3804,3706)(1147,-4314)
\put(3901,-2686){\makebox(0,0)[lb]{\smash{{\SetFigFont{5}{6.0}{\rmdefault}{\mddefault}{\updefault}{\color[rgb]{0,0,0}$I$}%
}}}}
\put(2776,-1486){\makebox(0,0)[lb]{\smash{{\SetFigFont{5}{6.0}{\rmdefault}{\mddefault}{\updefault}{\color[rgb]{0,0,0}$II$}%
}}}}
\put(3826,-1111){\makebox(0,0)[lb]{\smash{{\SetFigFont{5}{6.0}{\rmdefault}{\mddefault}{\updefault}{\color[rgb]{0,0,0}$\mathcal{H}^+$}%
}}}}
\put(4951,-2611){\makebox(0,0)[lb]{\smash{{\SetFigFont{5}{6.0}{\rmdefault}{\mddefault}{\updefault}{\color[rgb]{0,0,0}$\mathcal{I}$}%
}}}}
\put(4126,-1861){\rotatebox{315.0}{\makebox(0,0)[lb]{\smash{{\SetFigFont{5}{6.0}{\rmdefault}{\mddefault}{\updefault}{\color[rgb]{0,0,0}$v=v_0$}%
}}}}}
\put(3151,-3661){\makebox(0,0)[lb]{\smash{{\SetFigFont{5}{6.0}{\rmdefault}{\mddefault}{\updefault}{\color[rgb]{0,0,0}$v$}%
}}}}
\put(2251,-3586){\makebox(0,0)[lb]{\smash{{\SetFigFont{5}{6.0}{\rmdefault}{\mddefault}{\updefault}{\color[rgb]{0,0,0}$u$}%
}}}}
\end{picture}%

\]
In the familiar Eddington-Finkelstein double-null coordinate system $(U,v) \in \left(-\infty,\infty\right) \times \left(-\infty,\infty\right)$ the metric can be written as
\begin{align} \label{efc}
g_{SAdS} = - \left(1-\frac{2M}{r} + \frac{r^2}{l^2} \right) dU dv + r^2 \left(d\theta^2 + \sin^2 \theta d\varphi^2\right) \, ,
\end{align}
where $r=r\left(U,v\right)$ is the area radius, satisfying $r_U = -r_v = -\frac{1}{2}  \left(1-\frac{2M}{r} + \frac{r^2}{l^2} \right) $. This coordinate system covers only a subset of maximally extended Schwarzschild-AdS, which we denote by $I$ in the above Penrose diagram.
This region is commonly known as the domain of outer communications and may be described as $J^-\left(\mathcal{I}\right) \cap J^+\left(\mathcal{I}\right)$ once a suitable asymptotic notion of null infinity, $\mathcal{I}$, has been defined \cite{MITnotes}. The metric (\ref{efc}) is apparently singular where $r=r_{ASch}$, with $r_{ASch}$ being the unique real zero of the function $1-\frac{2M}{r} + \frac{r^2}{l^2}$.\footnote{We recall that $
r_{ASch} = p_+ + p_-  \textrm{ \ \ \ with \ \ \ \ } p_\pm = \left(Ml^2 \pm \sqrt{M^2 l^4 + \frac{l^6}{27}}\right)^\frac{1}{3}$
and hence that $pq = -\frac{l^2}{3}$ as well as $\frac{Ml^2}{r_{ASch}^3} = \frac{Ml^2}{2Ml^2 - \frac{r_{ASch}l^2}{3}}> \frac{1}{2}$.
\label{funoc}
}
It is well-known that the metric can be extended to values $r \leq r_{ASch}$ by a simple coordinate transformation and that the set where $r=r_{ASch}$ corresponds to two null-hypersurfaces $\mathcal{H}^+$ and $\mathcal{H}^-$ (horizons) intersecting in a bifurcate sphere (the center of the diagram above).
In the coordinate system (\ref{efc}), the future event horizon of the black hole, $\mathcal{H}^+$, is approached as $U \rightarrow \infty$. Note also that $\kappa=\gamma=\frac{1}{2}$ everywhere in the domain of outer communications.

Consider now a subregion of the maximally extended Schwarzschild spacetime, defined as the causal future of some ingoing null ray $N\left(v_0\right)$ (located at $v=v_0$ in the Eddington Finkelstein coordinates and with $inf_u N\left(v_0\right)=U_0$). lying to the future of the bifurcate sphere, as indicated by the shaded region in the figure above. 
As mentioned, one can introduce a regular coordinate system covering the entire causal future of $N\left(v_0\right)$ by simply changing the $u$-coordinate.  A convenient choice consists in setting
$u_0=U_0$ and
\begin{align}
 u=u_0  + \pi \cdot l - 2l \arctan \left(\frac{r}{l}\right) \textrm{ \ \ \ \ on $N\left(v_0\right)$}
\end{align}
while keeping the $v$-coordinate fixed (i.e.~still $\kappa = \frac{1}{2}$ globally). This defines a regular coordinate system in the shaded region. Note that the event horizon is now located at $u_{hoz} = u_0 + 2 l \left(\frac{\pi}{2} - \arctan \left(\frac{r_{ASch}}{l}\right)\right)$, while the $v$ coordinate still has infinite range. Note that in this coordinate system $u - v$ is not constant for past limit points of the null rays $N\left(v_0\right)$ (i.e.~along $\mathcal{I}$ in the geometric language)

In the next section, we will consider perturbations of the Schwarzschildean data on $v=v_0$ with respect to this regular coordinate system. 

\subsection{Perturbed Schwarzschild-AdS data}
We are going to specify initial data on a ray $N(v_0)$ as in \cite{gs:lwp}. There, a class of initial data, denoted $\mathcal{C}^{1+k}_{a,M}(\mathcal{N})$, was defined for any interval $\mathcal{N}=(u_0,u_1]$ and an explicit construction of such data was given. 
For convenience, we recall this construction in this section as this will enables to introduce the smallness condition on the matter fields.

Let us thus consider $N(v_0)$, which is of the form $(u_0,u_1]$, as our initial interval. 

We wish to keep the $u$-coordinate introduced in the previous section, which was regular at the horizon. Hence, we define $\bar{r}$ to be the solution of $\frac{\bar{r}_u}{1+\frac{\bar{r}^2}{l^2}} = -\frac{1}{2}$, with $\bar{r}$ tending to $\infty$ at $u_0$.
%
%
\subsubsection*{The free-data}
The free data then consists in a $C^2$-function $\bar{\phi} : N\left(v_0\right) \rightarrow \mathbb{R}$ satisfying the smallness bound
\begin{align} \label{pointwise}
 \bar{r}^{\frac{3}{2} + \frac{1}{2} s}\left( |\bar{\phi}| + | \bar{r} \frac{\bar{\phi}_u}{\bar{r}_u}|\right) + \Big| \bar{r}^\frac{7}{2} \frac{\partial_u \frac{\bar{\phi}_u}{\bar{r}_u}}{\bar{r}_u} \Big|  \leq \epsilon \textrm{\ \ \ everywhere on $N\left(v_0\right)$} \, ,
\end{align}
where $s=\min \left(\sqrt{9+8a},2\right)$, and in addition, being such that the combination
\begin{align}
\bar{\Phi} = \bar{r}^2 \left[ \bar{r} \ \partial_u \left(\frac{\bar{\phi}_{u}}{\bar{r}_u}\right) -4 \bar{\phi}_u - \frac{2a \bar{r}_u}{\bar{r}} \bar{\phi} \right]
\end{align}
is integrable, $\Pi\left(u\right) = \int_{u_0}^u \bar{\Phi} \left(\bar{u}\right) d\bar{u} < \epsilon$ for any $u \in N\left(v_0\right)$, and moreover, the bound
\begin{align} \label{addb}
\int_{N\left(v_0\right)} \left(\bar{\Phi}^2 \frac{\bar{r}^2}{\bar{r}_u} +\Pi^2 \bar{r}_u \right)  du < \epsilon^2
\end{align}
holds. Note that both (\ref{pointwise}) and (\ref{addb}) are independent of the choice of the $u$-coordinate. 

\subsubsection*{Deduced quantities}
From $\bar{\phi}$ we define the quantity $\bar{\varpi}$ as the unique $C^1$ solution of
\begin{align} 
\partial_u \bar{\varpi} &= 8 \pi \bar{r}^2 \frac{1-\frac{2\bar{\varpi}}{\bar{r}} + \frac{\bar{r}^2}{l^2}}{4r_u} \left(\partial_u \bar{\phi}\right)^2 + \frac{4\pi \bar{r}^2 a}{l^2} \bar{r}_u \bar{\phi}^2 \textrm{ \ \ \ , \ \ \ }
\lim_{u \rightarrow u_0} \bar{\varpi} \left(u\right) = M
\end{align}
and the $C^1$ quantity $\overline{r_v}$ as
\begin{align}
\overline{r_v} = \frac{1}{2} \left(1-\frac{2\bar{\varpi}}{\bar{r}} + \frac{\bar{r}^2}{l^2}\right) \exp \left( \int_{u_0}^u \frac{4\pi \bar{r}}{\bar{r}_u} \left(\partial_u \bar{\phi}\right)^2 du \right) \, .
\end{align}
Note that $\overline{r_v}$ is independent on the choice of $u$-coordinate on the data.
We also define the $C^1$ quantity
\begin{align}
\bar{\Omega}^2 = - \frac{4\bar{r}_u \overline{r_v}}{1- \frac{2\varpi}{r} + \frac{r^2}{l^2}} \, ,
\end{align}
and the shorthand $\bar{\kappa} = \frac{\overline{r_v}}{1-\mu}$, which both depend on the choice of coordinates .
Finally, we define the $C^1$ quantity $\overline{T\left(\phi\right)}$ as the unique solution of the ODE
\begin{align}
\partial_u \left(\bar{r} \bar{\kappa} \overline{T\left(\phi\right)}\right) = - \bar{r}\overline{r}_v \ \partial_u \frac{\bar{\phi}_{,u}}{\bar{r}_u}  + \bar{\phi}_u \left[ -2\overline{r}_v - 2\frac{\bar{\kappa} \bar{r}^2}{l^2} - \frac{2\bar{\kappa} \bar{\varpi}}{\bar{r}} + \frac{8 \pi \bar{r}^2 a \bar{\kappa} \bar{\phi}^2}{l^2}\right] - \frac{a \bar{\Omega}^2 \bar{r}}{2l^2} \bar{\phi} \nonumber 
\end{align}
with the boundary condition $\bar{r} \bar{\kappa} \overline{T\left(\phi\right)}=0$. It follows from the condition (\ref{addb}) and (\ref{pointwise}) that 
\begin{align} \label{later2}
\int_{u_0}^{u_1} \overline{r}^2  \left( \frac{\bar{r}^2}{\bar{r}_u} \left[\partial_u \overline{T\left(\phi\right)} \right]^2 + \overline{T\left(\phi\right)}^2 \bar{r}_u \right) du< C \epsilon^2.
\end{align}
and also that $|\overline{r}^\frac{3}{2}  \overline{T\left(\phi\right)}| < C \epsilon $, where $C>0$ is a constant depending only on $a$, $l$ and $M$.
Moreover, defining the $C^1$ quantity $\overline{\phi_v} = \bar{\kappa} \overline{T\left(\phi\right)} + \frac{\overline{r_v}}{\bar{r}_u} \partial_u \bar{\phi}$, we see that it satisfies
\begin{align} 
\left(\overline{\phi_v}\right)_u = - \frac{\bar{r}_u}{\bar{r}} \overline{\phi_v}  - \frac{\overline{{r}_v}}{\bar{r}} \bar{\phi}_u - \frac{\bar{\Omega}^2 a}{2l^2} \bar{\phi} \, .
\end{align}
Note that  (\ref{later2}) does not depend on the choice of $u$-coordinate.
\begin{definition} \label{defep}
An \underline{$\epsilon$-perturbed Schwarzschild-AdS data set} on $N\left(v_0\right)$ consists in a free function $\bar{\phi}: N\left(v_0\right) \rightarrow \mathbb{R}$ satisfying (\ref{pointwise}) and (\ref{addb}), together with the $C^1$ deduced quantities $(\bar{\varpi}, \bar{\Omega}, \overline{r_v})$ as defined above. In particular (\ref{later2}) holds for any $\epsilon$-perturbed Schwarzschild-AdS data set.
\end{definition}

\begin{remark}
In \cite{gs:lwp}, we constructed initial data with $-2\bar{r}_u=1-\frac{2M}{r}+\frac{r^2}{l^2}+o(\bar{r}^{-1})$.
Using the coordinate transformation
\begin{align}
\frac{du^\star}{du} = \frac{1+ \frac{r^2}{l^2}}{1 - \frac{2M}{r} + \frac{r^2}{l^2}} = 1 + \mathcal{O} \left(\frac{1}{r^3}\right)
\end{align}
near infinity, the $\epsilon$-perturbed data set becomes manifestly a $\mathcal{C}^{1+k}_{a,M}$ asymptotically AdS data set in the sense of \cite{gs:lwp}.
\end{remark}

\begin{remark}
In \cite{gs:lwp}, we choose the constant $s$ in \eqref{pointwise} to be $\min(\sqrt{9+8a},2)$, compared to $\min(\sqrt{9+8a},1)$. Indeed, the extra possible radial decay is not needed to prove local existence for the non-linear problem, as shown in \cite{gs:lwp}. As mentioned in \cite{gs:lwp}, the stronger decay nonetheless propagates, see for instance \cite{Holzegelwp}.
\end{remark}

Note also that by the uniqueness, specifying $\phi=0$ identically will yield (a subset of) the Schwarzschild spacetime as the maximum development.
\subsection{Maximum development and set-up} \label{se:tmd} \label{noco}
From \cite{gs:lwp}, it follows that any $\epsilon$-perturbed Schwarzschild-AdS data set admits a unique (up to diffeomorphism) maximal development. We refer to \cite{gs:lwp} for the precise statement of those results. 

If $\mathcal{Q} \subset \mathbb{R}^2$ denotes the quotient by the orbits of symmetry of the maximal development of some $\epsilon$-perturbed Schwarzschild-AdS data set and if $\lambda$ and $\nu$ are defined as in \eqref{eq:fon}, then let $\mathcal{R} \subset \mathcal{Q}$ denote the regular region, i.e.~the set of points such that $\lambda > 0$, $\nu < 0$.
From \cite{gs:lwp}, we have in particular:

\begin{proposition}
Let $(\mathcal{M},g)$ be the maximal development of some $\epsilon$-perturbed Schwarzschild-AdS data set. Let $\mathcal{Q}$ denote the quotient of the orbits. Then the past boundary of $\mathcal{Q}$ is a constant $v$ line $v=v_0$. 
Let the infimum of $u$ on $v=v_0$ be $u_0$. For $u > u_0$, let $N\left(u\right) \subset \mathcal{Q}$ denote the outgoing characteristic null-line $u=constant$ emanating from the initial data. Then, the set 
$$
 \{u>u_0 \  | \ \textrm{$N\left(u\right) \in \mathcal{R}$ and $r \rightarrow \infty$ along $N\left(u\right)$} \} \,
$$
is non-empty. Moreover, defining
\begin{align}
u_{\mathcal{H}} := \sup_{u> u_0} \{u \  | \ \textrm{$N\left(u\right) \in \mathcal{R}$ and $r \rightarrow \infty$ along $N\left(u\right)$} \}, \, 
\end{align}
as well as the subregion
\begin{align} \label{rhreg}
\mathcal{R}_{\mathcal{H}} := \mathcal{R} \cap \{ u_0 < u < u_{\mathcal{H}} \}  \textrm{ \ \ \ \ and \ \ \ } \overline{\mathcal{R}}_{\mathcal{H}} := \mathcal{R}_{\mathcal{H}} \cup N \left( u_{\mathcal{H}} \right),  \, 
\end{align}
%
$\mathcal{Q}$ contains a subset of the following form
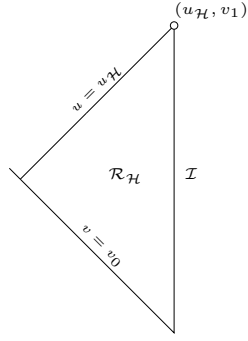
\[
\begin{picture}(0,0)%
\includegraphics{penAdS3.pstex}%
\end{picture}%
\setlength{\unitlength}{1816sp}%
\begingroup\makeatletter\ifx\SetFigFont\undefined%
\gdef\SetFigFont#1#2#3#4#5{%
  \reset@font\fontsize{#1}{#2pt}%
  \fontfamily{#3}\fontseries{#4}\fontshape{#5}%
  \selectfont}%
\fi\endgroup%
\begin{picture}(2412,4482)(739,-3673)
\put(1726,-2236){\rotatebox{315.0}{\makebox(0,0)[lb]{\smash{{\SetFigFont{5}{6.0}{\rmdefault}{\mddefault}{\updefault}{\color[rgb]{0,0,0}$v=v_0$}%
}}}}}
\put(1651,-661){\rotatebox{45.0}{\makebox(0,0)[lb]{\smash{{\SetFigFont{5}{6.0}{\rmdefault}{\mddefault}{\updefault}{\color[rgb]{0,0,0}$u=u_{\mathcal{H}}$}%
}}}}}
\put(2101,-1561){\makebox(0,0)[lb]{\smash{{\SetFigFont{6}{7.2}{\rmdefault}{\mddefault}{\updefault}{\color[rgb]{0,0,0}$\mathcal{R}_{\mathcal{H}}$}%
}}}}
\put(3001,689){\makebox(0,0)[lb]{\smash{{\SetFigFont{6}{7.2}{\rmdefault}{\mddefault}{\updefault}{\color[rgb]{0,0,0}$(u_{\mathcal{H}},v_1)$}%
}}}}
\put(3151,-1561){\makebox(0,0)[lb]{\smash{{\SetFigFont{6}{7.2}{\rmdefault}{\mddefault}{\updefault}{\color[rgb]{0,0,0}$\mathcal{I}$}%
}}}}
\end{picture}%

\]
In particular, first singularities cannot arise along $u=u_{\mathcal{H}}$, i.e.~the set \\$\{(u_{\mathcal{H}},v), v_0 \le v < v_1\}$ is included in $\mathcal{Q}$. 
Finally, there exists a global double-null coordinate system $(u,v)$ covering $\mathcal{R}_{\mathcal{H}}$ such that:
\begin{align}
\kappa = \frac{1}{2}  \textrm{ \  on $\mathcal{I}$, \ \ \ \ \ \  \ \ \ \ \ \ } \frac{-r_u}{1+\frac{r^2}{l^2}} = \frac{1}{2} \textrm{ \ on $v=v_0$}. \, 
\end{align}
\end{proposition}
\begin{proof}
By continuity, the data set contains a trapped surface, i.e. a point with $\lambda < 0$. Hence, Corollary A.2 of \cite{gs:lwp} applies. The existence of the $(u,v)$ coordinate system follows from a simple coordinate transformation.
\end{proof}
%
Note that under a change of null coordinates defined by $
\hat{u}=f(u)$, $\hat{v}=g(v)$,
the quantities $\Omega$, $\kappa$ and $\gamma$ transform as:
\begin{align}
\widehat{\Omega}=\frac{\Omega^2}{f'g'} \textrm{ \ \ \ \ \ , \ \ \ \ \ }
\hat{\kappa}=\frac{\kappa}{g'} \,  \textrm{ \ \ \ \ \  , \ \ \ \ \ }
\hat{\gamma}=\frac{\gamma}{f'} \, .
\end{align}

\subsubsection{Norms and Constants} \label{se:norm}
We define the following norms on $\mathcal{R}_{\mathcal{H}}$. For any point $\left(u,v\right) \in \mathcal{R}_{\mathcal{H}}$ let $u_{\mathcal{I}} \left(v\right)$ denote the $u$-coordinate of the point where the $v=const$-ray intersects null infinity $\mathcal{I}$.
\begin{align}
\| \psi \|^2_{H^1_{AdS}\left(u,v\right)} =  \int_{u_\mathcal{I}\left(v\right)}^u r^2 \left[\frac{r^2}{-r_u} \left( \partial_u\psi \right)^2   - r_u \psi^2 \right] \left(\bar{u},v\right)  d\bar{u}  \nonumber \\
+  \int_{v_0}^{v} r^2 \left[ \frac{1-\mu}{r_v} \left( \partial_v \psi \right)^2 + r_v \psi^2\right] \left(u,\bar{v}\right) d\bar{v} \, .
\end{align}
\begin{align}
\| \psi \|^2_{H^1_{AdS,deg}\left(u,v\right)} =  \int_{u_\mathcal{I}\left(v\right)}^u r^2 \left[\frac{1-\mu}{-r_u} \left( \partial_u\psi \right)^2 \left(\bar{u},v \right)  - r_u \psi^2 \right] \left(\bar{u},v\right)  d\bar{u}  \nonumber \\
+  \int_{v_0}^{v} r^2 \left[ \frac{1-\mu}{r_v} \left( \partial_v \psi \right)^2 + r_v \psi^2\right] \left(u,\bar{v}\right) d\bar{v} \, .
\end{align}
Note that both of these norms are independent of the choice of double null-coordinates. From \cite{gs:lwp}, it follows in particular that they are continuous in $\left(u,v\right)$. 
We also define spacetime energies capturing integrated decay:
\begin{align}
\mathbb{I} \left[\psi\right] \left(\mathcal{D} \right) = \int_{\mathcal{D}} \frac{1}{r^2} \left[ \frac{\left(\partial_u \psi \right)^2}{\gamma^2} + \frac{\left(\partial_v \psi \right)^2}{\kappa^2} + r^2 \psi^2 \right] \Omega^2 r^2  \left(\bar{u}, \bar{v}\right) d\bar{u} d\bar{v}\, , \nonumber
\end{align}
and also the non-degenerate integrated decay norm
\begin{align}
\overline{\mathbb{I}} \left[\psi\right] \left(\mathcal{D} \right) = \int_{\mathcal{D}} \frac{1}{r^2} \left[ \frac{r^4}{l^4}\frac{\left(\partial_u \psi \right)^2}{r_u^2} + \frac{\left(\partial_v \psi \right)^2}{\kappa^2} + r^2 \psi^2 \right] \Omega^2 r^2 \left(\bar{u}, \bar{v}\right) d\bar{u} d\bar{v} \, , \nonumber
\end{align} 
In practical applications, the region $\mathcal{D}$ is often going to be 
\begin{align}
D\left(u,v\right)=J^-\left(p=\left(u,v\right)\right) \cap J^+\left(N(v_0)\right), \,
\end{align}
where $N(v_0)$ is the initial ray.
Finally, we denote by $B_{M,l}$ a constant which does only depend on the fixed cosmological constant and the mass at infinity and by $B_{M,l,a}$ a constant which also depends on the fixed mass $a$.

\subsubsection{The constant $r$-curves $r_X$ and $r_Y$}
Define the point $q=\left(u_\mathcal{H},v_0\right)$. Clearly, since $\mathcal{R}_\mathcal{H}$ is part of the regular region, $r \geq r_{min} = r\left(u_\mathcal{H},v_0\right)$ holds in $\mathcal{R}_{\mathcal{H}}$. By the Raychaudhuri equation, a point on the initial data-ray $v=v_0$ at which $r_v<0$ cannot be part of $\mathcal{R}_H$. Since moreover in Schwarzschild $r_v<0$ holds for $r<r_{ASch}$, we have by the smallness assumption on the data the lower bound $r_{min} \geq \ras \left(1- C(\epsilon)\right)$, with $C(\epsilon) \rightarrow 0$ as $\epsilon \rightarrow 0$. Let $c>0$ be a small uniform constant (in particular, $c^\frac{1}{3}$ should still be much smaller than $a+\frac{9}{8}$) and define $r_Y$ as the unique real solution of
\begin{align} \label{choice}
1 - \frac{2M}{r_Y} + \frac{r_Y^2}{l^2} = c^\frac{1}{3} \, ,
\end{align}
Note that for $c$ small we have the estimate
\begin{align}
0 < r_Y - r_{ASch} \leq c^\frac{1}{3} \frac{2r_{ASch}}{1 + \frac{3 r_{ASch}^2}{l^2}} \, .
\end{align}
where $r_{ASch}$ is the $r$ value on the horizon in Schwarzschild-AdS. We choose $c$ so small that in particular
\begin{align} \label{shb}
\frac{2M\left(1-\sqrt{c}\right)}{r_Y^3} > \frac{1}{2} \, 
\end{align}
holds. Since this estimate is true for $c=0$ by footnote \ref{funoc}, this is possible by continuity.
Note that a-priori the curve $r=r_Y$ could lie outside of $\mathcal{R}_\mathcal{H}$, namely, if $r_{min}$ happens to be much larger than $r_{ASch}$.
In the same manner, we define a curve $r=r_X$ by solving
\begin{align}
1 - \frac{2M}{r_X} + \frac{r_X^2}{l^2} = d^\frac{1}{3} \, .
\end{align}
We assume $d>c$. As for $r_Y$, we have:
$$0 \le r_X - \ras \le B_{M,l} d^{1/3},$$
where $B_{M,l}$ is positive constant depending only on $M$ and $l$. By continuity, we can choose in particular $d$ so that
the following estimate holds:
\begin{equation} \label{rxchoice}
log \frac{r_{X}}{r_{min}} < \frac{1}{|a|} \, .
\end{equation}

%
%
%
%
%
%
%

%
%
%
%
%
%
%
%
\section{The main results} \label{se:tmt}
The main theorem can be found at the end of this section. We use this section to outline the sequence of propositions leading to the theorem. Some propositions are proven right away, while the proof of the three key propositions containing the crucial estimates is postponed to Sections \ref{mp1} and \ref{mp2}.

For the results below, recall the mass bound (\ref{confbound}), the definition of an $\epsilon$-perturbed Schwarzschild-AdS data set (Definition \ref{defep}) and that of the region $\mathcal{R}_{\mathcal{H}}$ associated with it, (\ref{rhreg}). 

Step 1 will be to establish \emph{uniform} bounds in the region $\mathcal{R}_\mathcal{H}$:
\begin{proposition}[Basic estimates] \label{uniformb}
There is an $\epsilon>0$ such that the solution arising from an $\epsilon$-perturbed Schwarzschild-data set satisfies the following estimate for $\left(u,v\right) \in \mathcal{R}_\mathcal{H}$:
\begin{align}
| \varpi - M |^\frac{1}{2} + | 2\kappa - 1 |^\frac{1}{2} +  | r^{\frac{3}{2}} \phi | + \Big| r^{\frac{3}{2}} \frac{\zeta}{\nu} \Big| + \|\phi\|_{H^{1}_{AdS}\left(u,v\right)} \nonumber \\
\leq B_{M,l,a} \left[ \| \phi \|_{H^{1}_{AdS}\left(u_\mathcal{H},v_0\right)} + \sup_{v=v_0} \Big| r^\frac{3}{2} \frac{\zeta}{\nu} \Big| \right] \, .
\end{align}
\end{proposition}
We remark that Proposition \ref{uniformb} actually holds for the entire Breitenlohner-Freedman range $a>-\frac{9}{8}$ and not only the range (\ref{confbound}).

\begin{proposition}[Improved and higher order bounds] \label{uniformb2}
Let
\begin{align} \label{idb}
\mathcal{N} \left[\phi\right] \left(v_0\right) = \Big[ \| \phi \|_{H^{1}_{AdS}\left(u_\mathcal{H},v_0\right)} &+ \|T\phi \|_{H^{1}_{AdS}\left(u_\mathcal{H},v_0\right)} \nonumber \\ 
+ \sup_{v=v_0} \Big| r^{\frac{3}{2}+\frac{s}{2}} \frac{\zeta}{\nu} \Big| &+ \sup_{v=v_0} \Big| r^\frac{5}{2} \frac{\partial_u \left(T\phi\right)}{\nu} \Big| +   \sup_{v=v_0}  \Big| r^\frac{7}{2} \frac{\partial_u \frac{\partial_u \phi}{r_u}}{r_u} \Big|  \Big]  < \infty
\end{align}
be a second order norm on the initial data.
There is an $\epsilon>0$ such that for any $\epsilon$-perturbed Schwarzschild-data set we have the following estimates for $\left(u,v\right) \in \mathcal{R}_\mathcal{H}$:
\begin{align} \label{sep1}
\Big| r^\frac{7}{2} \frac{\partial_u \frac{\partial_u \phi}{r_u}}{r_u} \Big| 
+\Big| r^\frac{5}{2} \frac{\partial_u \left(T\phi\right)}{\nu} \Big|  +  \|T\left(\phi\right)\|_{H^1_{AdS}\left(u,v\right)}  \leq  B_{M,l,a} \cdot \mathcal{N} \left[\phi\right] \left(v_0\right) \, ,
\end{align}
and, for any $\delta>0$ and $s=\min\left(\sqrt{9 + 8a},2-2\delta\right)$,
\begin{align} \label{sep2}
 | r^{\frac{3}{2}+\frac{s}{2}} \phi | + \Big| r^{\frac{3}{2}+\frac{s}{2}} \frac{\zeta}{\nu} \Big| + | r^{\frac{1}{2}+\frac{s}{2}} \phi_v |  \leq  C_\delta \cdot B_{M,l,a} \cdot \mathcal{N} \left[\phi\right] \left(v_0\right) \, ,
\end{align}
where $T\left(\phi\right) = \frac{1}{4\kappa} \partial_v \phi + \frac{1}{4\gamma} \partial_u \phi$.
\end{proposition}
Note that all these bounds are independent on the choice of $u$-coordinate.
\begin{proposition}(Integrated decay) \label{intdec3}
We have the integrated decay estimates
\begin{align} \label{ind1}
 \| \phi \|^2_{H^{1}_{AdS}\left(u,v\right) }+ \overline{\mathbb{I}}\left[\phi\right] \left(D\left(u,v\right)\right) \leq B_{M,l,a}  \| \phi \|^2_{H^{1}_{AdS}\left(u,v_0\right)}  \, ,
\end{align}
\begin{align} \label{ind2}
\overline{\mathbb{I}}\left[T\phi\right] \left(D\left(u,v\right)\right) \leq B_{M,l,a} \cdot \mathcal{N}^2 \left[\phi\right] \left(v_0\right) \, .
\end{align}
\end{proposition}
\begin{remark}
The proof of Proposition \ref{uniformb} will not require the construction of an integrated decay estimate. It exploits the redshift in terms of pointwise estimates, a characteristic feature of spherical symmetry \cite{DafRod}. On the other hand, the proof of both Propositions \ref{uniformb2} and \ref{intdec3} will require integrated decay for both $T$ and $T\phi$. As a gain, one can abandon the use of pointwise estimates entirely and obtain the pure $H^1$-estimate (\ref{ind1}). It is precisely in our construction of the integrated decay estimate that the restriction (\ref{confbound}) enters.

The estimate (\ref{sep2}) can be improved further by another commutation with $T$ in case that $\sqrt{9+8a}\geq2$, cf.~Remark \ref{moremore}.
\end{remark}

Step 2: The above estimates allow one to prove what is essentially the completeness of null-infinity:
\begin{proposition} \label{cni}
Let $v_m = \sup_{v\geq v_0} \{ v \ | \left(u_\mathcal{H},v\right) \in \mathcal{Q} \}$. We must have $v_{m} = \infty$.
\end{proposition}

\begin{proof}
Consider the family of curves of constant area radius $r$. In view of $r_u<0$ and $r_v>0$ holding in $\mathcal{R}_{\mathcal{H}}$, these curves are seen to be timelike in $\mathcal{R}_\mathcal{H}$ and to foliate $\mathcal{R}_{\mathcal{H}}$. Now, either none of these constant $r$-curves has future limit point $\left(u_{\mathcal{H}}, v_{m}\right)$, meaning that all of them intersect the horizon, or one of them, say $r=R$, has (and hence all later ones, $r>R$, as constant $r$ curves cannot intersect). In the latter case, we consider the infinite ``zig-zag"-curve as in the diagram below 
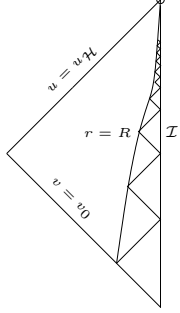
\[
\begin{picture}(0,0)%
\includegraphics{penAdS5.pstex}%
\end{picture}%
\setlength{\unitlength}{1816sp}%
\begingroup\makeatletter\ifx\SetFigFont\undefined%
\gdef\SetFigFont#1#2#3#4#5{%
  \reset@font\fontsize{#1}{#2pt}%
  \fontfamily{#3}\fontseries{#4}\fontshape{#5}%
  \selectfont}%
\fi\endgroup%
\begin{picture}(2187,4271)(889,-3673)
\put(3076,-1336){\makebox(0,0)[lb]{\smash{{\SetFigFont{6}{7.2}{\rmdefault}{\mddefault}{\updefault}{\color[rgb]{0,0,0}$\mathcal{I}$}%
}}}}
\put(1501,-1936){\rotatebox{315.0}{\makebox(0,0)[lb]{\smash{{\SetFigFont{5}{6.0}{\rmdefault}{\mddefault}{\updefault}{\color[rgb]{0,0,0}$v=v_0$}%
}}}}}
\put(1501,-736){\rotatebox{45.0}{\makebox(0,0)[lb]{\smash{{\SetFigFont{5}{6.0}{\rmdefault}{\mddefault}{\updefault}{\color[rgb]{0,0,0}$u=u_{\mathcal{H}}$}%
}}}}}
\put(1951,-1336){\makebox(0,0)[lb]{\smash{{\SetFigFont{5}{6.0}{\rmdefault}{\mddefault}{\updefault}{\color[rgb]{0,0,0}$r=R$}%
}}}}
\end{picture}%

\]
and observe that the $v$-length of each constant $u$-piece is uniformly bounded below. Namely, in view of the bound on $\kappa$ and the fact that $\frac{1-\mu}{r^2} \leq \frac{2}{l^2}$ to the right of the curve $r=R$ (for $R$ sufficiently large depending only on $M$ and $l$) we have for each constant $u$-piece $\mathcal{U}_i$
\begin{align}
\int_{\mathcal{U}_i} dv \geq \frac{l^2}{2} \int_{\mathcal{U}_i} \frac{\kappa \left(1-\mu\right)}{r^2} dv =   \frac{l^2}{2} \int_{\mathcal{U}_i} \frac{r_v}{r^2} dv\geq   \frac{l^2}{2R}
\end{align}
Since there are infinitely many $\mathcal{U}_i$ in the zig-zag curve ($r=R$ is timelike!), $v_{m}=\infty$ follows.

We turn to the first case (all constant $r$-curves intersecting the horizon and hence $\lim_{v \rightarrow v_{m}} r\left(u_{\mathcal{H}}, v\right) = \infty$). Assuming $v_m=V<\infty$ (otherwise, we are done) we will show that this contradicts the fact that $u=u_{\mathcal{H}}$ is the last $u$-ray along which $r=\infty$ can be reached. Pick $r=R$ very large, the corresponding curve intersecting $u_{\mathcal{H}}$ at $q= \left(u_{\mathcal{H}}, v_q\right)$, say.
In view of the assumptions and the uniform bounds on $\kappa$ and $\varpi$ of Propositions \ref{uniformb} and \ref{uniformb2}, we have $\frac{1-\mu}{r^2}>c>0$ in $\overline{\mathcal{R}_{\mathcal{H}}}$ for a constant $c$. Indeed, this is obvious in $\mathcal{R}_{\mathcal{H}} \cap \{r\geq R\}$ by computation, and immediate by compactness in $\overline{\mathcal{R}_{\mathcal{H}}} \cap \{r\leq R\}$, since $r_v = 0$ cannot hold anywhere in $\overline{\mathcal{R}_{\mathcal{H}}} \cap \{r\leq R\}$ (this would contradict that $r \rightarrow \infty$ along any $u=const$ ray in $\overline{\mathcal{R}_{\mathcal{H}}}$). It follows that $\gamma = -\frac{r_u}{1-\mu}$ is bounded on
the data $[u_0,u_{\mathcal{H}}] \times \{v_0\}$. Using the bound on $1-\mu$, $\kappa$ and $\phi$ one easily obtains, integrating (\ref{gammaR}) in $v$ that $\gamma$ is uniformly bounded in $\overline{\mathcal{R}_{\mathcal{H}}}$. By a change of $u$ coordinate, one achieves that $\gamma=\frac{1}{2}$ holds on $\mathcal{I}$. In the new coordinate system, one has that $u=v$ on $\mathcal{I}$, $\gamma_u=0$ on $\mathcal{I}$ and, integrating (\ref{gammaR}) from $\mathcal{I}$), the uniform bounds
\begin{align}
\Big| r^3 \left(\gamma - \frac{1}{2}\right) \Big| + | r^2 \gamma_u| < C \, .
\end{align}
With this established, all assumptions of the extension principle of Proposition 8.3 of \cite{gs:lwp} hold and we can extend the solution to a larger triangle, as depicted below.
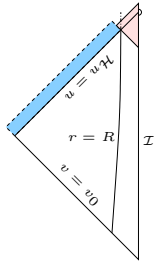
\[
\begin{picture}(0,0)%
\includegraphics{penAdS4.pstex}%
\end{picture}%
\setlength{\unitlength}{1460sp}%
\begingroup\makeatletter\ifx\SetFigFont\undefined%
\gdef\SetFigFont#1#2#3#4#5{%
  \reset@font\fontsize{#1}{#2pt}%
  \fontfamily{#3}\fontseries{#4}\fontshape{#5}%
  \selectfont}%
\fi\endgroup%
\begin{picture}(2337,4374)(739,-3673)
\put(3076,-1711){\makebox(0,0)[lb]{\smash{{\SetFigFont{5}{6.0}{\rmdefault}{\mddefault}{\updefault}{\color[rgb]{0,0,0}$\mathcal{I}$}%
}}}}
\put(1801,-961){\rotatebox{45.0}{\makebox(0,0)[lb]{\smash{{\SetFigFont{5}{6.0}{\rmdefault}{\mddefault}{\updefault}{\color[rgb]{0,0,0}$u=u_{\mathcal{H}}$}%
}}}}}
\put(1651,-2086){\rotatebox{315.0}{\makebox(0,0)[lb]{\smash{{\SetFigFont{5}{6.0}{\rmdefault}{\mddefault}{\updefault}{\color[rgb]{0,0,0}$v=v_0$}%
}}}}}
\put(1801,-1636){\makebox(0,0)[lb]{\smash{{\SetFigFont{5}{6.0}{\rmdefault}{\mddefault}{\updefault}{\color[rgb]{0,0,0}$r=R$}%
}}}}
\end{picture}%

\]
This contradicts the assumption that $u_{\mathcal{H}}$ is the last ray along which $r \rightarrow \infty$ can be reached.
\end{proof}
Step 3: We finally show that $r$ has to be bounded along the horizon, i.e.~we establish a Lorentzian Penrose inequality:
\begin{proposition} \label{penrose}
We have $\sup_{\mathcal{H}} r \leq r_Y$.
\end{proposition}
\begin{proof}
We will show that the curve $r=r_Y$ must lie entirely in $u < u_\mathcal{H}$, i.e.~in particular, it cannot cross the horizon $u=u_{\mathcal{H}}$. Clearly, by monotonicity of $r$ in $\mathcal{R}\cup \mathcal{A}$ this implies that $\sup_{\mathcal{H}} r \leq r_{Y}$ and hence the result. To establish the claim, we  suppose $r=r_Y$ crossed the horizon at some $v=v_i<\infty$ (the case that $r=r_Y$ lies completely outside $\mathcal{R}_{\mathcal{H}}$ is proven completely analogously). By Lemma \ref{omm} we have that $\frac{1-\mu}{r^2} \geq \frac{1}{8 r_Y^2} c^{\frac{1}{3}}$ holds on the ray $N\left(u_\mathcal{H}\right)$ to the future of $\left(u_\mathcal{H}, v_i\right)$. We then have, on the one hand,
\begin{align}
\int_{v_i}^v \frac{r_v}{r^2} dv = -\frac{1}{r\left(v\right)} + \frac{1}{r_Y}
\end{align}
which, in particular, is bounded above for all $v_i < v < \infty$. On the other hand,
\begin{align}
\int_{v_i}^v \frac{r_v}{r^2} dv = \int_{v_i}^v \frac{\kappa \left(1-\mu\right)}{r^2} dv \geq \frac{1}{4}\frac{1}{8 r_Y^2} c^{\frac{1}{3}} \left(v-v_i\right)  
\end{align}
which (in view of $v_{m} = \infty$ by Proposition \ref{cni}) can be made arbitrarily large by choosing $v$ sufficiently big. We have established a contradiction: The curve $r=r_Y$ cannot cross $u=u_{\mathcal{H}}$. In case that $r=r_Y$ lies entirely outside $\mathcal{R}_{\mathcal{H}}$ we use the same argument integrating along $u=u_\mathcal{H}$ from the data $v=v_0$ to some large $v$. 
\end{proof}

Note that we can choose $r_Y$ as close to $r_{ASch}$ as we desire by choosing $c$ and hence the initial data sufficiently small.

We summarize the statements of Propositions \ref{uniformb}, \ref{uniformb2}, \ref{cni} and \ref{penrose} as
\begin{theorem} \label{maintheorem}
Given an $\epsilon$-perturbed Schwarzschild-AdS data set on $N\left(v_0\right)$ in the sense of Definition \ref{defep}, its associated maximum development is a black hole spacetime with a regular future event horizon $\mathcal{H}^+$, and a complete null-infinity $\mathcal{I}$. Moreover, the estimates of Propositions \ref{uniformb}, \ref{uniformb2} and \ref{intdec3} hold for any $\left(u,v\right)$ on $J^+ \left(N\left(v_0\right)\right) \cap J^- \left(\mathcal{I}\right)$. This implies in particular that $\phi$ decays exponentially in $v$ on the latter set.
\end{theorem}
The last claim will be proven in section \ref{expvdec}.

We finally remark that the techniques of this paper (in particular, the integrated decay estimate of section \ref{se:vfid}) are naturally applied to the study of spherically symmetric solutions of the wave equation (\ref{wap}) on a fixed Schwarzschild-AdS background, yielding in particular the following result:
\begin{corollary} \label{lincol}
Spherically-symmetric solutions of (\ref{wap}) with the mass $a$ satisfying (\ref{confbound}) and $\left(\mathcal{M},g\right)$ a fixed Schwarzschild-AdS spacetime decay exponentially in the Eddington Finkelstein coordinate $v$ on the black hole exterior.
\end{corollary}

\section{Proof of Proposition \ref{uniformb}: Basic Estimates} \label{mp1}
%
%
%
%
%
%
Proposition \ref{uniformb} will be proven by a bootstrap.
\subsection{The bootstrap regions and the bootstrap assumptions}
We define, for $\tilde{u} \in \left[u_0,u_\mathcal{H}\right]$,
\begin{align} \label{bwh}
\widehat{\mathcal{B}} \left(\tilde{u}\right) = \mathcal{R}_{\mathcal{H}} \cap \{u_0 \leq u < \tilde{u} \} \, .
\end{align}
Note that $\mathcal{R}_\mathcal{H} = \widehat{\mathcal{B}} \left(u_{\mathcal{H}}\right)$.
Let
\begin{align}
u_{max} = \sup_u \Big( \textrm{ conditions (\ref{boot3})-(\ref{boot4a}) hold in $\widehat{\mathcal{B}}\left(u\right)$} \, \Big)
\end{align}

\begin{enumerate}
\item{ \textbf{Auxiliary bound:} \label{fi2}
\begin{align}
 |r^3 \phi^2 | < \frac{Ml^2}{8\pi|a|} \label{boot3} \, .
\end{align}
}
\item{
\textbf{Smallness of matter fields:} \label{last}
\begin{align}
\frac{4\pi \left(-a\right)}{l^2} \int_{v_0}^{v} d\bar{v} \,  \mathbf{1}_{\{r\leq r_Y\}} r^2 \, r_v \, \phi^2 \left(u,\bar{v}\right) & < M \cdot c \label{boot2} 
\end{align}
\begin{align} \label{boot2a}
2\pi \int_{v_0}^v  \mathbf{1}_{\{r\geq r_Y\}}  \frac{\phi_v^2}{\kappa}  r^2  \left(u,\bar{v}\right) d\bar{v} < M\sqrt{c} , 
\end{align}
\begin{align}
\frac{4\pi\left(-a\right)}{l^2} \int_{u_{\mathcal{I}}}^{u} d\bar{u} \,  \mathbf{1}_{\{r\leq r_Y\}} r^2 \, r_u \, \phi^2 \left(\bar{u},v\right) & < M \cdot c \label{boot4} 
\end{align}
\begin{align} \label{boot4a}
2\pi \int_{u_{\mathcal{I}}}^{u} d\bar{u} \mathbf{1}_{\{r\geq r_Y\}}  \frac{\phi_u^2}{\gamma}  r^2  \left(\bar{u},{v}\right) d\bar{u} < M\sqrt{c} , 
\end{align}
for any $\left(u,v\right)$ in $\widehat{\mathcal{B}}\left(u\right)$ where $\mathbf{1}_{\{...\}}$ is the indicator function.}
\end{enumerate}
$\phantom{X}$ \newline
Finally, define the bootstrap region $\mathcal{B}=\widehat{\mathcal{B}}\left(u_{\max}\right) \subset \mathcal{R}_{\mathcal{H}}$. \newline

We would like to prove that in fact $\mathcal{B}=\mathcal{R}_\mathcal{H}$. Now $\mathcal{B}$ is open in $\mathcal{R}_{\mathcal{H}}$ by continuity and also non-empty by Cauchy stability. Hence we are done if we could show that $\mathcal{B}$ is also closed in $\mathcal{R}_{\mathcal{H}}$. To do this, we assume $u_{max}<u_{\mathcal{H}}$ fixed (otherwise there is nothing to show) and prove that in $\overline{\mathcal{B}}\left(u_{max}\right)$ the bounds (\ref{boot3})-(\ref{boot4a}) can be improved.
\subsection{Overview of the argument}
We will show that the bootstrap assumptions imply that $M- c \leq \varpi \leq M$. This is done by exploiting Hardy inequalities both in the $u$- and the $v-$ direction but restricted to the region $r \geq r_Y$, as well as bootstrap assumption (\ref{boot2}) for the bad term in the region $r\leq r_Y$. Importantly, it will turn out that the Hardy inequalities do not require the entire good-signed derivative term (provided $a$ satisfies the Breitenlohner-Freedman bound). This enables us in turn to control $\phi$-flux trough characteristic lines by the mass difference and hence to improve (\ref{boot2a}) and (\ref{boot4a}) from $\sqrt{c}$ to $c$. With the improvement for the mass-flux, we invoke the redshift estimate and estimates from infinity to prove the pointwise bounds $|r^3 \phi^2| + |r^\frac{3}{2} \left( \frac{\phi_u}{r_u}\right)^2| < c$ everywhere, improving (\ref{boot3}). Finally, we improve (\ref{boot2}) and (\ref{boot4a}) using that the $r$ difference in the region $r\leq r_Y$ is $c^\frac{1}{3}$-small.
\subsection{Integrated positivity for $\varpi$ in $r \geq r_Y$}
An immediate consequence of the bootstrap assumptions is
\begin{lemma} \label{triv}
In the region $\mathcal{B}$ we have
\begin{align}
| \varpi - M | \leq 2M \sqrt{c}
\end{align}
\end{lemma}
which follows trivially from integration, and
\begin{lemma} \label{omm}
In the region $\mathcal{B} \cap \{ r \geq r_Y \}$ we have
\begin{align}
\frac{1- \mu}{r^2} \geq \frac{1}{8r_Y^2} c^\frac{1}{3} \, .
\end{align}
\end{lemma}
\begin{proof}
We write
\begin{align}
1- \frac{2\varpi}{r} + \frac{r^2}{l^2} = \left(1- \frac{2M}{r} + \frac{r^2}{l^2}\right) + \frac{2M-2\varpi}{r} \, .
\end{align}
Since $-\frac{2M}{r}$ is increasing in $r$ we can estimate it from below by $-\frac{2M}{r_Y}$. The mass difference is estimated by Lemma \ref{triv}. Hence, for $r\geq r_Y$, we have by (\ref{choice}),
\begin{align}
1- \frac{2\varpi}{r} + \frac{r^2}{l^2} \geq c^\frac{1}{3} - \frac{4M}{r_{Y}} \sqrt{c} + \frac{r^2-r_Y^2}{l^2}  \geq \frac{1}{2} c^\frac{1}{3} + \frac{r^2-r_Y^2}{l^2} \ , .
\end{align}
Dividing by $r^2$ we discard the the second term in the region $r \leq 2 r_Y$, while in $r \geq 2r_Y$ the second term is larger than $\frac{3}{4l^2}> \frac{1}{8r_Y^2} c^\frac{1}{3} $.
\end{proof}
\begin{lemma} \label{Lemma4}
For any $\aleph  < \frac{9}{8}$, the following inequality holds in $\overline{\mathcal{B}} \cap \{r\geq r_Y\}$:
\begin{align}
 \frac{2}{9} \frac{\aleph}{l^2} h^2\left(r-r_Y\right)^2 \leq \frac{\left(- r_u\right) r_v}{\Omega^2} \ \ \  , \ \ \textrm{where} \ \   h=1+\frac{r_Y}{r}+\left( \frac{r_Y}{r}\right)^2 \, .
\end{align}
\end{lemma}
\begin{proof}
Define
\begin{align}
\tilde{g} = + r_v \frac{\left(-r_u\right)}{\Omega^2} - \frac{2}{9}\aleph \hat{h}^2 \frac{\left(-a\right)}{l^2}\left(r-r_Y\right)^2
\end{align}
Note that $\tilde{g} \left(r_Y\right) \geq \frac{1}{32} c^\frac{1}{3}$ by Lemma \ref{omm}.

We first show that the same bound is valid on the initial data for $r\geq r_Y$. Again, this holds trivially where $r=r_Y$ intersects the data (call the $u$-coordinate of that point $u_Y$). The derivative in the $u$-direction satisfies
\begin{eqnarray}
\tilde{g}_u = r_u \Bigg[4\pi r \left(\frac{\phi_u}{r_u}\right)^2 \frac{r_u r_v}{\Omega^2} + \frac{\varpi}{2r^2} - \frac{2\pi a r\phi^2}{l^2} \nonumber \\
\hbox{}+\frac{r}{2l^2}\left(1-\frac{8}{9}\aleph - \frac{8}{9} \aleph \left(\frac{r_Y}{r}\right)^3\left(1-2 \left(\frac{r_Y}{r}\right)^3\right)\right) \Bigg] \, .
\end{eqnarray}
We have \begin{align}
\tilde{g}\left(u,v_0\right) = g\left(u_y,v_0\right) + \int_{u}^{u_y} \left(-\tilde{g}_u\right) \, ,
\end{align}
and we want to show $\tilde{g}\left(u,v_0\right) \geq 0$. Note that the third term in the square bracket has a good sign. We estimate
\begin{eqnarray}
\tilde{g}\left(u,v_0\right) \geq \frac{1}{32} c^\frac{1}{3} - \frac{\pi}{r_Y} \int^{u_Y}_{u} \frac{\phi_u^2}{\gamma} \nonumber \\
+ \int^{u_Y}_{u}\frac{-r_u r}{2l^2}\left(1-\frac{8}{9}\aleph - \frac{8}{9} \aleph \left(\frac{r_Y}{r}\right)^3\left(1-2 \left(\frac{r_Y}{r}\right)^3\right) + \frac{\left(M-\epsilon\right)l^2}{r^3} \right) du 
\end{eqnarray}
and observe that the second term can be estimated by the $H^1_{AdS,deg}$ norm on the data and is hence $\epsilon$-small. We conclude that the first line is already positive for sufficiently small data.
Thus, the lemma follows if we can show that the second line is positive. Let $A=\frac{8}{9} \aleph$, we have
\begin{align}
&\int^{u_Y}_{u}\frac{-r_u r}{2l^2}\bigg(1-\frac{8}{9}\aleph  -\frac{8}{9} \aleph \left(\frac{r_Y}{r}\right)^3\left(1-2 \left(\frac{r_Y}{r}\right)^3\right)+ \frac{\left(M-\epsilon\right) l^2}{r^3}\bigg)  du \nonumber \\
&= \frac{r_Y^2}{2l^2}\left[ -\frac{x^2}{2}(1-A)-\frac{A}{x}+\frac{A}{2x^4} + \frac{\left(M-\epsilon\right)l^2}{r_Y^3} \frac{1}{x}\right]_{x=r/r_Y}^{x=1}  \nonumber \\
&=\frac{r_Y^2}{2l^2} \left( -\frac{1}{2} +\left(\frac{r}{r_Y} \right)^2 \frac{1-A}{2}+A \frac{r_Y}{r}-\frac{A}{2}\left( \frac{r_Y}{r}\right)^4 + \frac{1}{2} \left(1- \frac{r_Y}{r}\right)  \right)\nonumber \\
&=\frac{r_Y^2}{2l^2} z\left(\frac{r}{r_Y},A\right)
\end{align}
where we used that $\frac{\left(M-\epsilon\right)l^2}{r_{Y}^3} > \frac{1}{2}$ holds by (\ref{shb}) to estimate the last term in the penultimate step.  
On the other-hand, we easily have:
\begin{lemma} \label{trivlem}
For any $0\leq A<1$, the function
\begin{align} \nonumber
z \left(x,A\right) =  \frac{1-A}{2} x^2 + \frac{A}{x} - \frac{A}{2x^4} - \frac{1}{2x} = \frac{1}{x^4} \left(\frac{1-A}{2}x^6 + \big(A- \frac{1}{2}\big)x^3 -\frac{A}{2} \right)
\end{align}
is non-negative in $[1,\infty)$.
\end{lemma}
\begin{proof}
Note that $z(x,A)$ is linear in $A$. For $A=0$, we have $z(x,0) x^4 =  \frac{x^6}{2}- \frac{x^3}{2}$ which is non-negative on $[1,\infty)$. For $A=1$, we have $z(x,1) x^4 = \frac{x^3}{2}-\frac{1}{2}$, which is also non-negative on $[1,\infty)$.
\end{proof}
To establish $\tilde{g} \geq 0$ in the entire region $\overline{\mathcal{B}} \cap \{r\geq r_Y\}$, it suffices to show that the bound is propagated in the $v$-direction. We compute
\begin{align} \label{derva}
\tilde{g}_v = -\pi r \frac{\phi_v^2 }{\kappa} + r_{,v} \frac{\varpi}{2r^2} - r_v \frac{2\pi a r \phi^2}{l^2} \nonumber \\
+r_v \frac{r}{2l^2}\left(1-\frac{8}{9}\aleph - \frac{8}{9} \aleph \left(\frac{r_Y}{r}\right)^3\left(1-2 \left(\frac{r_Y}{r}\right)^3\right)\right) \, .
\end{align}
In analogy to the previous case, we would like to show that $g\left(v\right) \geq \frac{1}{32} c^\frac{1}{3} + \int_{v_0}^v \mathbf{1}_{r\geq r_Y} \tilde{g}_v$ is positive. We note that the bad first term can now be estimated from the bootstrap assumption (\ref{boot2a}):
\begin{align}
\int_{v\left(r_Y\right)}^v  \pi r \frac{\phi_v^2 }{\kappa} d\bar{v} \leq \pi \frac{1}{r_Y} \int_{v\left(r_Y\right)}^v  \pi r^2 \frac{\phi_v^2 }{\kappa} \leq \pi \frac{M}{r_Y} \sqrt{c} \, .
\end{align}
In view of $\frac{1}{32}c^\frac{1}{3} - \pi \frac{M}{r_Y}c^\frac{1}{2}>0$ for sufficiently small $c$, we conclude that this term cannot drive $\tilde{g}$ to zero. To establish positivity of the integral for the other terms, we simply repeat the argument we followed in the $u$-direction reducing the problem to Lemma \ref{trivlem}.
\end{proof}

For the next Lemma, recall that $u_\mathcal{I}$ denotes the $u$-value where the $v=const$ curve intersects $\mathcal{I}$ and similarly $v_\mathcal{I}$ denotes the $v$-value where the $u=const$ curve intersects $\mathcal{I}$.
\begin{lemma} \label{Lemma2} \label{lem:vhardy}
For any $\aleph < \frac{9}{8}$ fixed, we have for any $\left(u,v\right) \in \mathcal{B} \cap \{ r\geq r_Y\}$, 
\begin{equation}
\int^{u}_{u_{\mathcal{I}}}  \, \frac{4\pi r^2 \aleph  }{l^2} \left(-r_u\right) \phi^2 \left(\bar{u},v\right)d\bar{u}  \leq \int^{u}_{u_\mathcal{I}} \, 8 \pi  r^2 \frac{r_v}{\Omega^2} (\partial_u \phi )^2\left(\bar{u},v\right)d\bar{u}  \, 
\end{equation}
and for any fixed $u=const$ curve in $\mathcal{B} $
\begin{equation}
\int_{v}^{v_\mathcal{I}}  \, \frac{4\pi r^2 \aleph }{l^2} \left(r_v\right) \phi^2 \left({u},\bar{v}\right)d\bar{v} \leq \int_{v}^{v_\mathcal{I}}  \, 8 \pi  r^2 \frac{-r_u}{\Omega^2} (\partial_v \phi )^2  \left({u},\bar{v}\right)d\bar{v} \, .
\end{equation}
\end{lemma}
\begin{proof}
We have by integration by parts:
$$
\int^u_{u_\mathcal{I}} \frac{4\pi r^2 \aleph }{l^2} \left(-r_u\right) \phi^2 du = -\frac{4\pi \aleph }{l^2}\frac{r^2h\left(r-r_{Y}\right)}{3} \phi^2 \bigg|^u_{u_\mathcal{I}} -\int^u_{u_\infty} \frac{8\pi \aleph r^2h\left(r-r_{Y}\right) }{3l^2} \phi \phi_u du \, , 
$$
where we recall $h=1+\frac{r_Y}{r}+\left( \frac{r_Y}{r} \right)^2$. Of the boundary terms on the right-hand side, one has a good (negative) sign, while the other vanishes by the decay of $\phi$ as $r \rightarrow \infty$. For the remaining term, we apply Cauchy-Schwarz:
\begin{align}
 \int^u_{u_\mathcal{I}} \frac{8\pi \aleph r^2 h \left(r-r_{Y}\right) }{3l^2} \phi \phi_u du \nonumber \\ 
 \le \frac{8 \pi \aleph}{3l^2}\left( \int^u_{u_\mathcal{I}} r^2 (-r_u) \phi^2 \right)^{1/2} \left( \int^u_{u_\mathcal{I}}\frac{r^2h^2\left(r-r_{Y}\right)^2}{-r_u} \phi_u^2 \right)^{1/2}, \nonumber
\end{align}
from which we deduce that:
\begin{eqnarray}
\frac{4 \pi \aleph}{l^2} \int^u_{u_\mathcal{I}}r^2(-r_u) \phi^2 du \le \frac{16 \pi}{9} \frac{\aleph }{ l^2}\int^u_{u_{\mathcal{I}}} \frac{r^2h^2\left(r-r_{Y}\right)^2}{-r_u} \phi^2_u du. \nonumber
\end{eqnarray}
An application of Lemma \ref{Lemma4} now yields the result. The inequality in the $v$-direction is similar.
\end{proof}
\begin{corollary} \label{ubm}
In the region $\mathcal{B}\cap \{r \geq r_Y\}$ the estimate $\varpi \leq M$ holds.
\end{corollary}
\begin{proof} We have
\begin{eqnarray}
\varpi-M=\int^{u}_{u_{\mathcal{I}}} \partial_u \varpi du=\int^{u}_{u_{\mathcal{I}}} du \, \left[ -8 \pi  r^2 \frac{r_v}{\Omega^2} (\partial_u \phi )^2+ \frac{4\pi r^2 \left(-a\right) }{l^2} \left(-r_u\right) \phi^2 \right] \nonumber
\end{eqnarray}
and by Lemma \ref{Lemma2} the right-hand side is negative.
\end{proof}

\subsection{Improving bootstrap assumptions (\ref{boot2a}) and (\ref{boot4a})}
From the conservation of Hawking mass we estimate for $\left(u,v\right) \in \mathcal{B}$,
\begin{align} \label{cones}
 \|\phi\|^2_{H^1_{AdS,deg}\left(u_{\mathcal{H}},v_0\right)} \geq \int_{u_0}^u  \left(2\pi \frac{\phi_u^2}{\gamma} - \frac{4\pi a}{l^2}\phi^2 r_u \right)  r^2 \left(\bar{u} ,v_0\right) d\bar{u} =\nonumber \\
\int_{u_\mathcal{I}}^u  \left[2\pi \frac{\phi_u^2}{\gamma} - \frac{4\pi a}{l^2}\phi^2 r_u \right]  r^2 \left(\bar{u} ,v\right) d\bar{u} + \int_{v_0}^v  \left[2\pi \frac{\phi_v^2}{\kappa} + \frac{4\pi a}{l^2}\phi^2  r_v \right] r^2 \left(u ,\bar{v}\right) d\bar{v} 
\nonumber \\
\geq \frac{1}{2}\left(a+\frac{9}{8}\right)  \|\phi\|^2_{H^1_{AdS,deg}\left(u,v\right)} + \frac{9}{8}\frac{4\pi}{l^2} \int_{u_\mathcal{I}}^u \mathbf{1}_{r\leq r_Y} \phi^2 r_u r^2 \left(\bar{u} ,v\right) d\bar{u} \nonumber \\ - \frac{9}{8}\frac{4\pi}{l^2} \int_{v_0}^v  \mathbf{1}_{r\leq r_Y} \phi^2 r_v r^2 \left(u ,\bar{v}\right) d\bar{v}  \, ,
 \end{align}
where we used the Hardy inequalities established in Lemma \ref{lem:vhardy}.
Using the bootstrap assumptions (\ref{boot2}) and (\ref{boot4}) for the terms in $r\leq r_Y$, we establish
\begin{corollary} \label{fienb}
For any $\left(u,v\right) \in \mathcal{B}$ we have
\begin{align}
\| \phi \|_{H^1_{AdS,deg}\left(u,v\right)} \leq C_a \cdot M \cdot \sqrt{c}
\end{align}
with $C_a$ a uniform constant, depending only on how close $a$ is to the BF-bound.
\end{corollary}
This improves in particular bootstrap assumptions (\ref{boot2a}) and (\ref{boot4a}) and also shows that the overall mass-difference is $\sqrt{c}$-small.

\subsection{Estimating $\phi$ in $r \geq r_X$}
Next we derive a pointwise smallness bound for $\phi$ in $r \geq r_X$ (not $r_Y$!), by integrating in $u$ from infinity:

\begin{lemma} \label{phibound}
For all $(u,v) \in \mathcal{B} \cap \{ r \geq r_X \}$, we have
\begin{eqnarray}
| r^\frac{3}{2} \phi(u,v) | \le B_{M,l} \cdot d^{-1/6} \|\phi\|_{H^1_{AdS,deg}\left(u,v\right)} \leq  B_{M,l,a} \cdot \sqrt{c}  \, .
\end{eqnarray}
\end{lemma}
\begin{proof}
Integrating out from infinity (where $\phi$ vanishes) we find
\begin{align}
|\phi \left(u_{r \geq r_X},v\right)| &\leq 0 + \Big| \int du \phi_u \Big| \leq \sqrt{\int du \zeta^2 \frac{\lambda}{\Omega^2}} \sqrt{\int du \frac{4}{r^2 \left(1-\mu\right)} \left(-r_u\right) } \nonumber \\
&\leq \frac{2\sqrt{M}r_X^2}{r^\frac{3}{2}}d^{-\frac{1}{6}} \|\phi\|_{H^1_{AdS,deg}\left(u,v\right)} \nonumber \, ,
\end{align}
where we used the upper bound on the $v$-flux as well as the estimate $8 r_X^2 \frac{1-\mu}{r^2} \ge d^{1/3}$ (cf.~Lemma \ref{omm}), which holds in the region where $r \ge r_X$.
\end{proof}

Note that on $r=r_Y$ we would only obtain $c^\frac{1}{3}$-smallness, as the bad $\left(1-\mu\right)^{-1}$-weight would bring in an inverse $c$-smallness. 
\subsection{The red-shift effect: $\frac{\zeta}{\nu}$ estimate in $r\leq r_X$}
Recall $\zeta = r\phi_u$. We can write the wave equation as
\begin{align} \label{rsevol}
\partial_v \left( \frac{\zeta}{\nu} \right)= -\phi_v +\frac{2r\kappa a \phi}{l^2}
-\frac{\zeta}{\nu} \left[2 \kappa \frac{\varpi}{r^2} + \frac{2 \kappa r}{l^2} - 8\pi r \frac{a}{l^2} \kappa \phi^2 \right].
\end{align}

\begin{lemma} \label{zetahozbound}
For any  $\left(u,v\right) \in \mathcal{B} \cap \{ r \leq r_X \}$ we have
\begin{align} \label{taes}
\Big| \frac{\zeta}{\nu}\Big| + |\phi| \leq B_{M,l}\left( \sup_{D\left(u,v\right)} \|\phi\|_{H^1_{AdS,deg}\left(u,v\right)} + \sup_{v=v_0}\Big| r^\frac{3}{2} \frac{\zeta}{\nu}\Big| \right) \leq B_{M,l} \cdot \sqrt{c} \, .
\end{align}
\end{lemma}

\begin{proof}
Let us denote the redshift weight
\begin{align} \label{rsweight}
\rho = 2\kappa \left[ \frac{\varpi}{r^2} + \frac{  r}{l^2} -  4\pi r \frac{a}{l^2}  \phi^2 \right]
\end{align}
Note that $\frac{\rho}{\kappa} > \left(\frac{2r_{min}}{l^2} + \frac{M}{r_Y^2}\right)$.\footnote{Due to the cosmological term, we actually have a global redshift at work. In the asymptotically-flat case, the strength of the redshift degenerates at infinity, in view of the absence of that term. We will exploit this good term which grows in $r$ (``the redshift at infinity") later in the estimates near infinity.} Integrating (\ref{rsevol}) we find
\begin{align} \label{tte}
 \frac{\zeta}{\nu} \left(u,v\right)= \left( \frac{\zeta}{\nu} \left(u,v_0\right)  \right)\cdot \exp \left(\int_{v_0}^v -\rho \left(u, \bar{v} \right) d\bar{v}\right) \nonumber \\
 +   \int_{v_0}^v d\bar{v} \left[ \exp \left(-\int_{\bar{v}}^v \rho \left(u, \hat{v} \right) d\hat{v}\right)\left(-\phi_v +\frac{2r\kappa a \phi}{l^2}\right) \left(u,\bar{v} \right) \right]  \, .
\end{align}
Let us study the inhomogeneous term. For the $\phi_v$-term we need to estimate
\begin{align} \label{fioc}
\left| \int_{v_0}^v d\bar{v} \left[ \exp \left(-\int_{\bar{v}}^v \rho \left(u, \hat{v} \right) d\hat{v}\right) \phi_v\right]\right| \nonumber \\
\leq \sqrt{\int_{v_0}^v  d\bar{v} \frac{1}{r^2} \kappa \cdot \exp \left(-2\int_{\bar{v}}^v \rho \left(u, \hat{v} \right) d\hat{v}\right)} \sqrt{\int_{v_0}^v \frac{\phi_v^2}{\kappa} r^2 \left(u,\bar{v}\right) d\bar{v} } \, .
\end{align}
The second square root can be controlled from the energy, while the first can be estimated by a constant:
\begin{align}
\int_{v_0}^v  d\bar{v} \frac{1}{r^2} \kappa \cdot \exp \left(-2\int_{\bar{v}}^v \rho \left(u, \hat{v} \right) d\hat{v}\right) 
= \int_{v_0}^v  d\bar{v} \frac{\kappa}{2 r^2 \rho} \partial_{\bar{v}}  \exp \left(-2\int_{\bar{v}}^v \rho \left(u, \hat{v} \right) d\hat{v}\right) \nonumber \, ,
\end{align}
and we can take out the supremum of $\frac{\kappa}{2 r^2 \rho}$ because the derivative of the exponential has a positive sign, i.e.~the integrand is positive everywhere. This finally yields
\begin{align}
\int_{v_0}^v  d\bar{v} \frac{1}{r^2} \kappa \cdot \exp \left(-2\int_{\bar{v}}^v \rho \left(u, \hat{v} \right) d\hat{v}\right) \nonumber \\
\leq \sup \left(\frac{\kappa}{2r^2 \rho}\right) \cdot \left[1 - \exp\left(-2\int_{v_0}^v \rho \left(u, \bar{v} \right) d\bar{v} \right)\right] \nonumber \\
\leq \frac{1}{4} \cdot \sup \left[ \left(\varpi  + \frac{r^3}{l^2} -  4\pi r^3 \frac{a}{l^2}  \phi^2 \right)^{-1}\right] \, .
\end{align}
The $\phi$-term in (\ref{tte}) is more delicate because a smallness bound on $\phi$ is not available close to the horizon. The only thing we have at our disposal is that $\int \phi^2 r_v dv$ is controlled by the energy (which is not immediately useful because $r_v$ may be very small in the region under consideration). The idea is to integrate the inhomogeneity by parts, since a $v$-derivative falling on $r$ will generate the required factor of $r_v$ . We write
\begin{align} \label{expip}
 \int_{v_0}^v d\bar{v} \left[ \exp \left(-\int_{\bar{v}}^v \rho \left(u, \hat{v} \right) d\hat{v}\right)\left(\frac{2r\kappa a \phi}{l^2}\right) \left(u,\bar{v} \right) \right]  \nonumber \\
 =
 \int_{v_0}^v d\bar{v} \frac{2r\kappa a \phi}{l^2 \rho} \partial_{\bar{v}} \left[ \exp \left(-\int_{\bar{v}}^v \rho \left(u, \hat{v} \right) d\hat{v}\right) \left(u,\bar{v} \right) \right]  \, .
\end{align}
When we integrate by parts, the boundary terms that arise are one on data (which is $\epsilon$-small by assumption) and one at $\left(u,v\right)$, which is
\begin{align}
\Bigg|\left(\frac{a r}{2l^2} \frac{1}{\left[ \frac{\varpi}{r^2} + \frac{r}{l^2} -  4\pi r \frac{a}{l^2}  \phi^2 \right]}\right) \phi \left(u,v\right) \Bigg| \leq \frac{|a|}{2}  |\phi| \,  .
\end{align}
To analyze the volume term we compute
\begin{align} \label{compt}
\partial_v \left(\phi \frac{2r\kappa}{\rho}\right) =  \frac{2r\kappa}{\rho} \left(1 + \frac{2r\kappa}{\rho}\frac{8\pi  a}{l^2 } \phi^2\right) \phi_v \nonumber \\ + \left( \frac{2r \kappa}{\rho}\right)^2  \left(\frac{3\varpi}{r^4} - \frac{4\pi a}{rl^2} \phi^2\right)  \phi \ r_v 
- \left( \frac{2r\kappa}{\rho}\right)^2 \ 2\pi \frac{\phi}{r} \frac{\phi_v^2}{\kappa} \, .
\end{align}
Note again that the factor $\frac{2r\kappa}{\rho}$ is both bounded above and below. Hence the term proportional to $\phi_v$ can (after using (\ref{boot3})) be estimated as before (cf.~(\ref{fioc})). For the term proportional to $r_v$ we use Cauchy-Schwarz 
\begin{align} 
\left| \int_{v_0}^v d\bar{v} \left[ \exp \left(-\int_{\bar{v}}^v \rho \left(u, \hat{v} \right) d\hat{v}\right) \phi r_v \right]\right| \nonumber \\
\leq \sqrt{\int_{v_0}^v  d\bar{v} \frac{1}{r^2} \kappa \left(1-\mu\right) \cdot \exp \left(-2\int_{\bar{v}}^v \rho \left(u, \hat{v} \right) d\hat{v}\right)} \sqrt{\int_{v_0}^v \phi^2 r^2 \ r_v  \left(u,\bar{v}\right) d\bar{v} } \, ,
\end{align}
recovering the $H^1_{AdS,deg}$-norm. Finally, for the cubic term in (\ref{compt}) we apply the pointwise auxiliary bootstrap assumption (\ref{boot3}) to $\phi$ and estimate the remainder by the \emph{square} of the $H^1_{AdS,deg}$-norm. Since this norm itself is $\sqrt{c}$ small by Corollary \ref{fienb}, we have $\|\phi\|^2_{H^1_{AdS, deg}} \leq M \sqrt{c} \|\phi\|_{H^1_{AdS, deg}}$.

We summarize that (\ref{tte}) finally turns into the estimate 
\begin{align} \label{gbt}
\Big|\frac{\zeta}{\nu} \left(u,v\right) \Big|  \leq B_{M,l}\left[ \sup_{D\left(u,v\right)} \|\phi\|_{H^1_{AdS,deg}\left(u,v\right)} + \sup_{v=v_0}\Big| r^\frac{3}{2} \frac{\zeta}{\nu}\Big| \right] + \frac{|a|}{2} \cdot |\phi|\left(u,v\right) \, ,
\end{align}
valid for $\left(u,v\right) \in \mathcal{B} \cap \{r\leq r_X\}$. Note that the pointwise norm on $r^\frac{3}{2} \frac{\zeta}{\nu}$ controls in particular the $\phi$-term picked up on the data in the integration by parts.
From this we derive an estimate for $\phi$ by integrating from the fixed $r=r_X$-curve towards the horizon:
\begin{align}
|\phi \left(u,v\right)| \leq |\phi \left(u_{r_X},v\right)| + \int_{u_{r_X}}^u \Big|\frac{\zeta}{\nu}\Big| \frac{\left(-r_u\right)}{r} d\bar{u} 
\end{align}
leads, after applying Lemma \ref{phibound} and (\ref{gbt}), to
\begin{align}
 \sup_{D\left(u,v\right) \cap \{r\leq r_X\}}  |\phi \left(u,v\right)| \leq B_{M,l} \  \sup_{D\left(u,v\right)} \|\phi\|_{H^1_{AdS,deg}\left({u},{v}\right)} + \log \frac{r_X}{r_{min}} \Bigg[  \nonumber \\ B_{M,l}\left( \sup_{D\left(u,v\right)} \|\phi\|_{H^1_{AdS,deg}\left({u},{v}\right)} + \sup_{v=v_0}\Big| r^\frac{3}{2} \frac{\zeta}{\nu}\Big| \right) + \frac{|a|}{2}   \sup_{D\left(u,v\right) \cap \{r\leq r_X\}}  |\phi \left(u,v\right)| \Bigg] , \nonumber 
\end{align}
from which the estimate (\ref{taes}) follows for $\phi$ recalling our choice (\ref{rxchoice}).
Revisiting the estimate (\ref{gbt}), we obtain the same bound for $r^\frac{3}{2} \frac{\zeta}{\nu}$.
\end{proof}
\subsection{Improving assumptions (\ref{boot3}), (\ref{boot2}) and (\ref{boot4})} \label{fiim}
Note that assumption (\ref{boot3}) has already been improved in view of Proposition \ref{zetahozbound} and Lemma \ref{phibound}.

Using the global pointwise smallness bound for $\phi$ in the region $r \leq r_X$ established in Lemma \ref{zetahozbound}, we can improve both (\ref{boot2}) and (\ref{boot4}), using that the $r$- difference in the region $r\leq r_Y$ is $c^\frac{1}{3}$ small. For (\ref{boot2}):
\begin{align}
\frac{4\pi |a|}{l^2} \int_{v_0}^{v} d\bar{v} \,  \mathbf{1}_{\{r\leq r_Y\}} r^2 \, r_v \, \phi^2 \left(u,\bar{v}\right) \nonumber \\
\leq \frac{4\pi |a|}{l^2} \sup_{r\leq r_Y} |r^2 \phi^2| \left(r_Y-r_{min}\right) \leq  B_{M,l} c^\frac{4}{3} < \frac{1}{2} M \cdot c \, .
\end{align}
Assumption (\ref{boot4}) is improved completely analogously. 
This improves the last of the bootstrap assumptions and we conclude that $\mathcal{B}=\mathcal{R}_{\mathcal{H}}$. In the final subsection we explain how this implies the estimates of Proposition \ref{uniformb}.
\subsection{Conclusions} \label{conclusio}
Note first that inserting the estimate of Lemma \ref{zetahozbound} into (\ref{cones}) actually yields 
\begin{align} \label{onc1}
 \|\phi\|^2_{H^1_{AdS,deg}\left(u,v\right)} \leq B_{M,l,a} \left[ \|\phi\|^2_{H^1_{AdS,deg}\left(u_{\mathcal{H}},v_0\right)} + \sup_{v=v_0}\Big| r^\frac{3}{2} \frac{\zeta}{\nu}\Big| \right] ,
\end{align}
after exploiting the $c^\frac{1}{3}$-smallness. From the general estimate
\begin{align} \label{convt}
\sup_{D\left(u,v\right)}  \|\phi\|_{H^1_{AdS}\left(u,v\right)} \leq B_{M,l} \left[ \sup_{D\left(u,v\right)} \|\phi\|_{H^1_{AdS,deg}\left(u,v\right)} + \sup_{D\left(u,v\right) \cap \{r\leq r_X\}} \Big| r^\frac{3}{2} \frac{\zeta}{\nu}\Big| \right] ,
\end{align}
and Lemma \ref{zetahozbound}, we conclude that (\ref{onc1}) also holds for the non-degenerate norm on the left-hand side.

To finally conclude Proposition \ref{uniformb} we need the pointwise bound on $\kappa$ and a bound for $r^{\frac{3}{2}} \frac{\zeta}{\nu}$ in the region $r\geq r_X$. For $\kappa$, we integrate (\ref{kappaR}) to obtain
$$
\kappa(u,v)=\frac{1}{2} \exp \left( \int^u_{u_\mathcal{I}} \frac{4 \pi r}{\nu}(\partial_u \phi)^2 du \right) \, .
$$
Clearly, $\kappa(u,v) \le \frac{1}{2} $ globally. To derive a lower bound on $\kappa$ in the region $r \geq r_X$ we estimate
\begin{align}
\kappa(u,v) \geq \frac{1}{2} \exp \left(-\sup_{r\geq r_X} {\frac{1}{r \left(1-\mu\right)}} \int^u_{u_\infty} r^2 \frac{\lambda}{\Omega^2}(\partial_u \phi)^2 du \right) \, .
\end{align}
Now since in $r\geq r_X$ we have $8 r_Y^2 \frac{\left(1-\mu\right)}{r^2} \geq d^\frac{1}{3}$, we can conclude that
\begin{align}
\kappa \left(u,v\right) \geq \frac{1}{2} \exp \left(- \frac{8}{r_Y} d^{-\frac{1}{3}} \cdot M \|\phi\|^2_{H^1_{AdS,deg}\left(u,v\right)} \right) \textrm{ \ \ \ in $\mathcal{R}_{\mathcal{H}} \cap \{r\geq r_X\}$.}
\end{align}
From $r\leq r_X$ we continue to integrate up to the boundary of $\mathcal{R}_{\mathcal{H}}$, now using the bound on $\frac{\zeta}{\nu}$ established in Lemma \ref{zetahozbound}:
\begin{align}
\kappa(u,v)  \geq \frac{1}{2}\exp \left(- B_{M,l} \|\phi\|^2_{H^1_{AdS,deg}\left(u,v\right)} \right) \exp \left( \int^u_{u_{r_Y}} \frac{4 \pi r}{\nu}(\partial_u \phi)^2 du \right) \nonumber \\ 
\geq \frac{1}{2}\exp \left(- B_{M,l} \sup_{\mathcal{R}_\mathcal{H}} \|\phi\|^2_{H^1_{AdS,deg}\left(u,v\right)} \right) \exp \left( -\sup_{\mathcal{R}_\mathcal{H} \cap \{ r\leq r_X\} } \Big|\frac{\zeta}{\nu} r^\frac{3}{2} \Big|^2 \int^u_{u_{r_Y}} \frac{-\nu}{r^4} du \right) \nonumber \\
 \geq \frac{1}{2} -  B_{M,l}\left( \sup_{\mathcal{R}_\mathcal{H}} \|\phi\|^2_{H^1_{AdS,deg}\left(u,v\right)} + \sup_{v=v_0}\Big| r^\frac{3}{2} \frac{\zeta}{\nu}\Big|^2 \right)  \, , \nonumber
\end{align}
where we used Lemma \ref{zetahozbound} in the last step. 

For the $r$-weighted estimated for $\frac{\zeta}{\nu}$ we prove
\begin{lemma} \label{erb}
In the entire region $\mathcal{B}$ we have
\begin{align} 
\Big| r^\frac{3}{2} \frac{\zeta}{\nu}\Big| + |\phi| \leq B_{M,l}\left( \sup_{\mathcal{B}} \|\phi\|_{H^1_{AdS,deg}\left(u,v\right)} + \sup_{v=v_0}\Big| r^\frac{3}{2} \frac{\zeta}{\nu}\Big| \right) \, .
\end{align}
\end{lemma}
\begin{proof}
In view of Lemma \ref{phibound} and Lemma \ref{zetahozbound}, we only need to derive the bound for $\frac{\zeta}{\nu}r^\frac{3}{2}$ in the region $r\geq r_X$. We compute
\begin{equation} \label{uf}
\partial_v \left( r^n \frac{\zeta}{\nu} \right)= - r^n \phi_v + \frac{2 \kappa a \phi }{l^2} r^{n+1}-r^n \frac{\zeta}{\nu} \left[ - \frac{n\lambda}{r}+ \frac{2 \kappa \varpi}{r^2} + \frac{2 \kappa r}{l^2} - \frac{8\pi ra}{l^2} \kappa \phi^2 \right] \, ,
\end{equation}
and observe that
\begin{align}
-n \frac{\lambda}{r}+2 \kappa \frac{\varpi}{r^2} + \frac{2 \kappa r}{l^2} - 8\pi r \frac{a}{l^2} \kappa \phi^2 = -n \frac{\kappa \left(1-\mu\right)}{r}+ \frac{2 \kappa \varpi}{r^2} + \frac{2 \kappa r}{l^2} - 8\pi r \frac{a}{l^2} \kappa \phi^2 \nonumber \\
= \kappa \left[ 2 \left(n+1\right) \frac{\varpi}{r^2} + \frac{r}{l^2} \left(2-n\right) - \frac{n}{r} - 8\pi r \frac{a}{l^2} \phi^2 \right] \nonumber \, .
\end{align}
Choosing $n=\frac{3}{2}$ we see that we gain an exponential decay factor for large $r$. We integrate (\ref{uf}) in $v$ from $r=r_X$ (where we already established the bound, Lemma \ref{zetahozbound}) or from the initial data to any point in $\mathcal{R}_H$, which leads to the estimate
\begin{align}
\Big| r^\frac{3}{2} \frac{\zeta}{\nu} \left(u,v\right) \Big| \leq   \sup_{v=v_0}\Big| r^\frac{3}{2} \frac{\zeta}{\nu}\Big|  + B_{M,l} \cdot \frac{1}{\sqrt{r}} \int_{v_0}^v \mathbf{1}_{\{r \geq r_X\}} \sqrt{r} \left[ r^\frac{3}{2} |\phi_v| + \frac{2\kappa}{l^2} |\phi| r^\frac{5}{2} \right] \nonumber \\
\leq   \sup_{v=v_0}\Big| r^\frac{3}{2} \frac{\zeta}{\nu}\Big|  + B_{M,l} \cdot \frac{1}{\sqrt{r}} \|\phi\|_{H^1_{AdS,deg}\left(u,v\right)} \sqrt{\int_{v_0}^v \mathbf{1}_{\{r \geq r_X\}} r_v dv} \nonumber
\end{align}
where we used both Cauchy-Schwarz and that $r_v \geq \frac{1}{B_{M,l}} r^2$ holds in $r\geq r_X$. The desired estimate follows. -- We remark that later we will improve this estimate considerably using commutation.
\end{proof}

\section{Vectorfields and an integrated decay estimate} \label{se:vfid}
\subsection{Vectorfield identities}
Let $X=X^u \left(u,v\right) \partial_u + X^v \left(u,v\right) \partial_v$ be a vectorfield and $\mathfrak{f}\left(u,v\right)$ a function. 
We have the following formula for the deformation tensor of $X$:
\begin{equation}
 2\phantom{}^{(X)}\pi^{ab} = g^{ac} \partial_c X^b + g^{bd} \partial_d X^a + g^{ac} g^{bd} g_{cd,f} X^f \, \, ,
\end{equation}
and hence the following non-vanishing components:
\begin{align}
\pi^{uu} = -\frac{2}{\Omega^2} \partial_v X^u \textrm{ \ \ \ , \ \ \ } \pi^{vv} = -\frac{2}{\Omega^2} \partial_u X^v \, ,
\end{align}
\begin{align}
\pi^{uv} = -\frac{1}{\Omega^2} \left(\partial_v X^v + \partial_u X^u\right) - \frac{2}{\Omega^2} \left(\frac{\Omega_u}{\Omega} X^u + \frac{\Omega_v}{\Omega} X^v\right)  \, ,
\end{align}
\begin{align}
\pi^{AB} = \frac{1}{r} g^{AB} \left(r_u X^u + r_v X^v\right) \, .
\end{align}
Let $\psi$ satisfy the equation\footnote{In applications, $\psi$ will be $T\phi$ and hence $\mathfrak{q}\left[T\phi\right]$ the error arising from commutation with the vectorfield $T$.}
\begin{align}
\Box_g \psi -\frac{2a}{l^2} \psi = \mathfrak{q}\left[\psi\right] \, .
\end{align}
Then the energy momentum tensor
\begin{align}
\mathbb{T}_{\mu \nu} \left[\psi\right] = \partial_\mu \psi \partial_\nu \psi-\frac{1}{2}g_{\mu \nu} (\partial \psi )^2-\frac{a}{l^2} \psi^2 g_{\mu \nu}
\end{align}
satisfies
\begin{align}
\nabla^\mu \mathbb{T}_{\mu \nu} \left[\psi\right]= \left(\nabla_\nu \psi \right)  \mathfrak{q}\left[\psi\right] \, .
\end{align} 
For future use we collect its components
\begin{align}
\mathbb{T}_{uu} \left[\psi\right] = \left(\partial_u\psi \right)^2 \textrm{ \ \ \ , \ \ \ } \mathbb{T}_{vv} \left[\psi\right] = \left(\partial_v\psi \right)^2 \textrm{ \ \ \ , \ \ \ } \mathbb{T}_{uv} \left[\psi\right] =\frac{a \Omega^2}{2l^2} \psi^2 \, ,
\end{align}
\begin{align}
g^{AB} \mathbb{T}_{AB} \left[\psi\right] = \frac{4}{\Omega^2}\partial_u \psi \partial_v \psi -2 \frac{a}{l^2} \psi^2 \, .
\end{align}
We want to make use of the following multiplier identity
\begin{align} \label{Xid}
\nabla^\mu {J}^{X,\mathfrak{f}}_\mu \left[\psi\right] = K^{X,\mathfrak{f}} \left[\psi\right]
\end{align}
where 
\begin{align}
{J}^{X,\mathfrak{f}}_\mu \left[\psi\right] &= \mathbb{T}_{\mu \nu}\left[\psi\right] X^\nu + \mathfrak{f} \psi \nabla_\mu \psi - \frac{1}{2} \psi^2 \nabla_\mu \mathfrak{f}  \nonumber \\
K^{X,\mathfrak{f}} \left[\psi\right] &= \mathbb{T}_{\mu \nu} \left[\psi\right] \pi^{\mu \nu} + X\left(\psi\right) \mathfrak{q}\left[\psi\right]+ \mathfrak{f} \left[g^{\mu \nu} \partial_\mu \psi \partial_\nu \psi \right] + \left(-\frac{1}{2} \Box \mathfrak{f} + \frac{2a}{l^2} \mathfrak{f} \right) \psi^2 + \mathfrak{f} \psi q\left[\psi\right] \nonumber
\end{align}
We compute
\begin{align} \label{mfvf}
K^{X,\mathfrak{f}} \left[\psi\right]  = - \frac{2}{\Omega^2}\left(\partial_v X^u\right) \left(\partial_u \psi\right)^2 - \frac{2}{\Omega^2}\left(\partial_u X^v\right) \left(\partial_v \psi\right)^2 \nonumber \\
+ \left(\partial_u \psi\right) \left(\partial_v \psi\right) \left[ \frac{4r_u}{\Omega^2 r} X^u + \frac{4r_v}{\Omega^2 r} X^v - \frac{4}{\Omega^2} \mathfrak{f} \right] + \left(X\left[\psi\right] + \mathfrak{f} \psi\right) q\left[\psi\right] \nonumber \\
- \frac{a}{l^2} \psi^2 \left[ -2\mathfrak{f} + \frac{l^2}{2a} \Box \mathfrak{f} + \partial_u X^u + \left(2\frac{r_u}{r} + 2 \frac{\Omega_u}{\Omega} \right) X^u + \partial_v X^v + \left(2\frac{r_v}{r} + 2 \frac{\Omega_v}{\Omega} \right) X^v\right] 
\end{align}
We finally remark that the identity (\ref{Xid}) will typically be integrated over the diamond shaped region $D\left(u,v\right)$ defined in section \ref{noco}.
\subsection{The vectorfield $T\left[\phi\right]= \frac{1}{4\kappa}\partial_v \phi + \frac{1}{4\gamma} \partial_u \phi$}
The non-vanishing components of the deformation tensor of $T$ are
\begin{align}
{}^{(T)}\pi^{uu} = \frac{1}{2\Omega^2} \frac{\gamma_v}{\gamma^2} = 8\pi \ r \left(\frac{\partial_v \phi}{\Omega^2}\right)^2  \textrm{  \ ,  \ } {}^{(T)}\pi^{vv} = \frac{1}{2\Omega^2} \frac{\kappa_u}{\kappa^2} = - 8\pi \ r \left(\frac{\partial_u\phi}{\Omega^2}\right)^2 \, .\nonumber
\end{align}
This is because
\begin{align}
2{}^{(T)}\pi^{uv} = \frac{1}{2\Omega^2} \frac{\kappa_v}{\kappa^2} + \frac{1}{2\Omega^2} \frac{\gamma_u}{\gamma^2} - \frac{1}{2} \left(\frac{2}{\Omega^2}\right)^2  X \left(\Omega^2\right) = 0
\end{align}
\begin{align}
{}^{(T)}\pi^{AB} = \frac{2}{r^3} X\left(r\right) g^{AB} = 0 \, ,
\end{align}
with the last two identities following from
\begin{align}
 \frac{2}{\Omega^2} \frac{\kappa_v}{\kappa^2} = -\frac{2}{\Omega^2} \partial_v \left(\frac{1}{\kappa}\right) = \frac{2}{\Omega^2} \partial_v \left(\frac{4r_u}{\Omega^2}\right) = \frac{8r_{uv}}{\Omega^4} + \frac{2}{\Omega^4} \frac{1}{\kappa} \partial_v \Omega^2 \, ,
\end{align}
\begin{align}
 \frac{2}{\Omega^2} \frac{\gamma_u}{\gamma^2} = -\frac{2}{\Omega^2} \partial_u \left(\frac{1}{\gamma}\right) = -\frac{2}{\Omega^2} \partial_u \left(\frac{4r_v}{\Omega^2}\right) = -\frac{8r_{uv}}{\Omega^4} + \frac{2}{\Omega^4} \frac{1}{\gamma} \partial_u \Omega^2 \, ,
\end{align}
on the one hand, and $T\left(r\right)=0$ on the other. It follows that ${}^{(T)}\pi^{ab} \mathbb{T}_{ab}  \left[\phi \right] = 0$, which means that for our non-linear system $T$ is not Killing, but nevertheless leads to a conservation law in view of the identity (\ref{Xid}) becoming
\begin{align}
\nabla^a \left(\mathbb{T}_{ab}\left[\phi\right] T^b\right) = 0 \, .
\end{align}
Inspecting the boundary-terms generated by $T$ it becomes apparent that the Hawking mass is a potential for the energy fluxes of the vectorfield $T$ through a hypersurface.
\subsection{An integrated decay estimate for $\phi$}
Recall the norms defined in section \ref{noco}. We define the fluxes
\begin{align}
\mathbb{F} \left(u,v\right) &= \| {\phi}\|^2_{H^1_{AdS}\left(u,v\right)} + \| {\phi}\|^2_{H^1_{AdS}\left(u,v_0\right)} \nonumber \\
\mathbb{F}_{deg} \left(u,v\right) &= \| {\phi}\|^2_{H^1_{AdS,deg}\left(u,v\right)} + \| {\phi}\|^2_{H^1_{AdS,deg}\left(u,v_0\right)}
\end{align}
\begin{proposition} \label{intdecest}
For any $\left(u,v\right) \in \mathcal{R}_{\mathcal{H}}$ we have
\begin{align} 
\int_{u_{\mathcal{I}}}^u \frac{\left(\partial_u \phi\right)^2}{-\nu} \left(\bar{u},v\right) d\bar{u} + \int_{v_0}^v \kappa \phi^2 \left(\bar{u},v\right) d\bar{v} + \overline{\mathbb{I}} \left[\phi\right] \left(D\left(u,v\right) \right) \nonumber \\
 \leq B_{M,l} \left[ \| {\phi}\|^2_{H^1_{AdS,deg}\left(u,v\right)} + \| {\phi}\|^2_{H^1_{AdS}\left(u,v_0\right)} \right]
\end{align}
\end{proposition}
Note that the boundary terms on the left are almost equal to the non-degenerate $H^1_{AdS}$-norm, \emph{except} that their $r$-weight at infinity is weaker.

Proposition \ref{intdecest} will follow from the sequence of propositions proven in the remainder of this subsection. 
We apply (\ref{Xid}) with a vectorfield $X$ for which
\begin{align}
X^u = -\frac{r_v}{\Omega^2} \mathfrak{F}\left(r\right) \textrm{ \ \ \ and \ \ \ } X^v = -\frac{r_u}{\Omega^2} \mathfrak{F}\left(r\right) \, ,
\end{align}
so that
\begin{align}
\partial_v X^u = 4\pi r \frac{\left(\partial_v \phi\right)^2}{\Omega^2} \mathfrak{F} \left(r\right) - \frac{r_v}{\Omega^2} \mathfrak{F}^\prime\left(r\right) r_v \, ,
\end{align}
\begin{align}
\partial_u X^v = 4\pi r \frac{\left(\partial_u \phi\right)^2}{\Omega^2} \mathfrak{F} \left(r\right) - \frac{r_u}{\Omega^2} \mathfrak{F}^\prime\left(r\right) r_u \, ,
\end{align}
\begin{align}
\partial_u X^u + 2\frac{\Omega_u}{\Omega} X^u + \partial_v X^v + 2\frac{\Omega_v}{\Omega} X^v=  -2\frac{r_v r_u}{\Omega^2} \mathfrak{F}^\prime\left(r\right) -2\frac{r_{vu}}{\Omega^2} \mathfrak{F}\left(r\right) \, .
\end{align}
We split
\begin{align}
K^{X,\mathfrak{f}} \left[\phi\right]  = K^{X,\mathfrak{f}}_{main} \left[\phi\right] + K_{error}^{X,\mathfrak{f}} \left[\phi\right] \, ,
\end{align}
where
\begin{align} \label{KXfm}
K^{X,\mathfrak{f}}_{main} \left[\phi\right] = 2 \mathfrak{F}^\prime \left(r\right) \left[ \frac{r_v}{\Omega^2} \partial_u \phi +  \frac{r_u}{\Omega^2} \partial_v \phi \right]^2 + \nonumber \\
+ \left(\partial_u \phi\right) \left(\partial_v \phi\right) \left[ -\frac{4r_u r_v}{\Omega^2 \Omega^2} \left(\mathfrak{F}^\prime + \frac{2}{r} \mathfrak{F} \right) - \frac{4}{\Omega^2} \mathfrak{f} \right] \nonumber \\
- \frac{a}{l^2} \phi^2 \left[ -2\mathfrak{f}  -2\frac{r_v r_u}{\Omega^2} \left(\mathfrak{F}^\prime + \frac{2}{r} \mathfrak{F} \right) -2\frac{r_{vu}}{\Omega^2} \mathfrak{F}\left(r\right)+ \frac{l^2}{2a} \Box \mathfrak{f}\right] 
\end{align}
and
\begin{align}
K^{X,\mathfrak{f}}_{error} \left[\phi\right] = - \frac{16}{\Omega^4} \pi r \left(\partial_u \phi\right)^2 \left(\partial_v \phi\right)^2 \mathfrak{F}\left(r\right) \, .
\end{align}
We choose
\begin{align}
\mathfrak{f} = -\frac{r_u r_v}{\Omega^2} \left(\mathfrak{F}^\prime + \frac{2}{r} \mathfrak{F} \right) = \frac{1}{4} \left(1-\mu\right) \left(\mathfrak{F}^\prime + \frac{2}{r} \mathfrak{F} \right) \, ,
\end{align}
so that
\begin{align} \label{KXmain}
K^{X,\mathfrak{f}}_{main} \left[\phi\right] = 2 \mathfrak{F}^\prime \left(r\right) \left( \frac{r_v}{\Omega^2} \partial_u \phi +  \frac{r_u}{\Omega^2} \partial_v \phi \right)^2 \nonumber \\
- \frac{a}{l^2} \phi^2 \left[ -2\frac{r_{vu}}{\Omega^2} \mathfrak{F}\left(r\right)+ \frac{l^2}{2a} \Box_g \left(-\frac{r_u r_v}{\Omega^2} \left(\mathfrak{F}^\prime + \frac{2}{r} \mathfrak{F} \right)  \right) \right] \, .
\end{align}
We would like to find a bounded, monotonically increasing function, $\mathfrak{F}\left(r\right)$, since this will make the boundary term in the multiplier identity controllable by the energy and will, in addition, give the derivative term in (\ref{KXmain}) a sign. Exploiting the remaining freedom in $\mathcal{F}$ to make the square bracket in (\ref{KXmain}) globally positive is difficult (if not impossible). However, the next proposition shows that the zeroth order term can be absorbed by the derivative term for a simple choice of $\mathfrak{F}$.

\begin{proposition}
We have for any $a\geq -1$ the estimate
\begin{align}
 \int_{D\left(u,v\right)}  \frac{1}{r^6} \left( \frac{1}{4\gamma} \phi_u - \frac{1}{4\kappa}\phi_v \right)^2 \frac{\Omega^2}{2} r^2 \ du \ dv \ d\sigma_{S^2} \leq B_{M,l} \cdot \mathbb{F}_{deg} \left(u,v\right) \, . \nonumber
\end{align}
\end{proposition}

\begin{proof}
Apply the identity (\ref{Xid}) with $\mathfrak{F}\left(r\right) = - \frac{1}{r^2}$ (and hence $\mathfrak{f}=0$). We first look at the boundary terms. We have
\begin{align}
\int_{D\left(u,v\right)}  \nabla^\mu {J}^{X,\mathfrak{f}}_\mu \left[\phi\right] = \int_{v_0}^v \left(\mathbb{T}_{vv} X^v + \mathbb{T}_{uv} X^u \right)  \left(u,\bar{v} \right) d\bar{v} d\sigma_{S^2} \nonumber \\
\int_{u_\mathcal{I}}^u \left(\mathbb{T}_{uu} X^u + \mathbb{T}_{uv} X^v  \right)  \left(\bar{u},v\right) d\bar{u} d\sigma_{S^2}  - \int_{u_0}^u \left(\mathbb{T}_{uu} X^u + \mathbb{T}_{uv} X^v  \right)   \left(\bar{u},v_0 \right) d\bar{u} d\sigma_{S^2} \nonumber 
\end{align}
because the boundary term on $\mathcal{I}$ vanishes. It is not hard to see that
\begin{align}
\Big| \int_{D\left(u,v\right)}  \nabla^\mu {J}^{X,\mathfrak{f}}_\mu \left[\phi\right] \Big| \leq  B_{M,l} \cdot\mathbb{F}_{deg} \left(u,v\right)  \nonumber \, .
\end{align}
We turn to the spacetime term. Observe that
\begin{align} \label{choicea1}
\frac{2a}{l^2}\phi^2 \frac{r_{vu}}{\Omega^2} \mathfrak{F}\left(r\right) =  \frac{a}{l^2 r^2} \phi^2 \left(\frac{r}{l^2} + \frac{\varpi}{r^2} - \frac{4\pi r a \phi^2}{l^2} \right) \, .
\end{align}
This term has the same sign as $a$ since the bracket is positive. (Hence for $a>0$ we are done immediately.) Moreover we can write both
\begin{align}
\frac{a}{l^2 r^2} \phi^2 \left(\frac{r}{l^2} + \frac{\varpi}{r^2} - \frac{4\pi r a \phi^2}{l^2} \right) = \frac{a}{l^2 r^2} \left(\frac{1}{2 r_u} \partial_u \left(1-\mu\right) - \frac{8\pi r \ r_v}{r_u \Omega^2} \left(\partial_u \phi\right)^2 \right) \phi^2 \nonumber
\end{align}
and
\begin{align}
\frac{a}{l^2 r^2} \phi^2 \left(\frac{r}{l^2} + \frac{\varpi}{r^2} - \frac{4\pi r a \phi^2}{l^2} \right) = \frac{a}{l^2 r^2} \left(\frac{1}{2 r_v} \partial_v \left(1-\mu\right) - \frac{8\pi r \ r_u}{r_v \Omega^2} \left(\partial_v \phi\right)^2 \right) \phi^2 \nonumber
\end{align}
Integrating this zeroth order terms yields (in view of $\sqrt{g} = \frac{\Omega^2}{2} r^2$)
\begin{align}
\int_{D\left(u,v\right)} \frac{a}{l^2 r^2} \phi^2 \left(\frac{r}{l^2} + \frac{\varpi}{r^2} - \frac{4\pi r a \phi^2}{l^2} \right) = \nonumber \\
\frac{1}{2} \int_{D\left(u,v\right)}  \frac{a}{l^2} \phi^2\left(-\kappa \cdot \partial_u \left(1-\mu\right) - \frac{4\pi r \ r_v}{r_u} \left(\partial_u \phi\right)^2 \right) du \ dv \ d\sigma_{S^2} \nonumber \\
\frac{1}{2} \int_{D\left(u,v\right)}  \frac{a}{l^2}\phi^2 \left(\gamma \cdot \partial_v \left(1-\mu\right) - \frac{4\pi r \ r_u}{r_v} \left(\partial_v \phi\right)^2 \right) du \ dv \ d\sigma_{S^2} \, .
\end{align}
Integrating the first term in the second line by parts, we see that if the derivative hits the $\kappa$ it will cancel with the second term in that line. Similarly for the third line and the derivative falling on $\gamma$. This means that
\begin{align}
\int_{D\left(u,v\right)} \frac{a}{l^2 r^2} \phi^2 \left(\frac{r}{l^2} + \frac{\varpi}{r^2} - \frac{4\pi r a \phi^2}{l^2} \right) \frac{\Omega^2}{2} r^2 \ du \ dv  \nonumber \\
= \int_{D\left(u,v\right)} \frac{4a}{l^2} \gamma \kappa \left(1-\mu\right)  \phi \left( \frac{1}{4\gamma} \phi_u - \frac{1}{4\kappa}\phi_v \right) du \ dv  \nonumber \\
- \int_{u_0}^u \frac{a}{4 l^2} \phi^2 \left(-r_u\right) \left(\bar{u},v_0 \right) d\bar{u} + \int_{u_{\mathcal{I}}}^u \frac{a}{4 l^2} \phi^2  \left(-r_u\right)  \left(\bar{u},v \right) d\bar{u} - \int_{v_0}^v \frac{a}{4 l^2} \phi^2 \left(r_v\right) \left(u, \bar{v}\right) d\bar{v} \nonumber
\end{align}
since again the boundary term on $\mathcal{I}$ vanishes. Clearly, the boundary terms are manifestly controlled by $\mathbb{F}_{deg} \left(u,v\right)$. For the remaining spacetime term we note, using $xy \leq \frac{x^2}{2} + \frac{y^2}{2}$ combining the estimate with the previous identity
\begin{align}
 \int_{D\left(u,v\right)} \frac{4|a|}{l^2} \gamma \kappa \left(1-\mu\right)  \phi \left( \frac{1}{4\gamma} \phi_u - \frac{1}{4\kappa}\phi_v \right)du \ dv   \nonumber \\
 \leq  \int_{D\left(u,v\right)}  \frac{|a|}{2l^2} \phi^2 \left(\frac{r}{l^2} + \frac{\varpi}{r^2} - \frac{4\pi r a \phi^2}{l^2} \right) \frac{\Omega^2}{2} \ du \ dv \  + \nonumber \\
 \int_{D\left(u,v\right)}  \frac{32|a|}{l^2 } \frac{\gamma^2 \kappa^2 \left(1-\mu\right)^2}{\Omega^4} \left(\frac{r}{l^2} + \frac{\varpi}{r^2} - \frac{4\pi r a \phi^2}{l^2} \right)^{-1} \left( \frac{1}{4\gamma} \phi_u - \frac{1}{4\kappa}\phi_v \right)^2 \frac{\Omega^2}{2}  du \ dv \nonumber 
\end{align}
and using that  $\frac{\gamma^2 \kappa^2 \left(1-\mu\right)^2}{\Omega^4}=\frac{1}{16}$ we obtain
\begin{align}
\frac{1}{2}\int_{D\left(u,v\right)} \frac{|a|}{l^2 r^2} \phi^2 \left(\frac{r}{l^2} + \frac{\varpi}{r^2} - \frac{4\pi r a \phi^2}{l^2} \right) \frac{\Omega^2}{2} r^2 \ du \ dv \leq \nonumber \\
  \int_{D\left(u,v\right)}   \frac{2|a|}{r} \left(1 + \frac{\varpi l^2}{r^3} - 4\pi a \phi^2 \right)^{-1}  \left( \frac{1}{4\gamma} \phi_u - \frac{1}{4\kappa}\phi_v \right)^2 \frac{\Omega^2}{2}  du \ dv + B_{M,l} \cdot\mathbb{F}_{deg} \left(u,v\right) \, .
 \nonumber
\end{align}
Finally, in view of the estimate
\begin{align}
1 - a  \left(1 + \frac{\varpi l^2}{r^3} - 4\pi a \phi^2 \right)^{-1}  \geq \frac{M l^2}{2r^3} \, ,
\end{align}
which holds for all $0\geq a \geq -1$, we conclude
\begin{align}
\int_{D\left(u,v\right)} K^{X,f} \left[\phi\right] \geq \int_{D\left(u,v\right)}  \frac{2M l^2}{r^4} \left( \frac{1}{4\gamma} \phi_u - \frac{1}{4\kappa}\phi_v \right)^2 \frac{\Omega^2}{2}  du dv - B_{M,l} \cdot\mathbb{F}_{deg} \left(u,v\right) \nonumber ,
\end{align}
as $K^{X,0}_{error}\left[\phi\right]$ has a good sign. The proposition follows.
\end{proof}

We have established an integrated decay estimate for a certain (radial) derivative. From here it is relatively straightforward to derive estimates for the other derivative and the zeroth order term, and to finally improve the $r$-weights near infinity and eliminate the degeneration near the horizon. Repeating the proof above with a slightly different $r$-weight, we can derive the following more general Hardy inequality:
\begin{corollary} \label{heco}
Set $f= r^n \left(\frac{2r}{l^2} + \frac{2\varpi}{r^2} - \frac{8\pi r a \phi^2}{l^2} \right) + n r^{n-1} \left(1-\mu\right)$. We have, for $0\leq n \leq 2$, the estimate
\begin{align}
\int_{D\left(u,v\right)} f \phi^2 \Omega^2 du dv \leq \int_{D\left(u,v\right)} \Omega^2 \frac{r^{2n}}{f} \left(\frac{\phi_u}{\gamma} - \frac{\phi_v}{\kappa} \right)^2 du dv + B_{M,l} \cdot\mathbb{F}_{deg} \left(u,v\right) \nonumber
\end{align}
\end{corollary}
\begin{proof}
Write $f= \frac{1}{r_u} \partial_u \left(r^n \left(1-\mu\right) \right) - r^n \frac{16\pi r r_v}{r_u \Omega^2} \left(\partial_u \phi\right)^2$ and analogously for the $v$-derivative. Then repeat the proof above. Note that $n=0$ yields the previous case.
\end{proof}
\begin{proposition}
We have for any $a\geq -1$ the estimate
\begin{align}
 \int_{D\left(u,v\right)} \left(\phi^2 + \frac{1}{r^2} \frac{1}{\gamma^2} \left(\partial_u \phi\right)^2 +\frac{1}{r^2} \frac{1}{\kappa^2} \left(\partial_v \phi\right)^2    \right) \frac{\Omega^2}{2} r^2 \ du \ dv \leq B_{M,l} \cdot\mathbb{F}_{deg} \left(u,v\right) \nonumber
\end{align}
\end{proposition}
\begin{proof}
Step 1: We first establish that there is a bounded $C^3$ function $\mathcal{F}\left(r\right)$ with $|r^5 \mathcal{F}\left(r\right)| + |r^6 \mathcal{F}^\prime\left(r\right)| \leq B_{M,l}$ such that
$$\mathcal{O}(\frf)= -\frac{1}{8}\square_g\left( (1-\mu)\left(\frf^\prime+\frac{2\frf}{r}\right)\right) -\frac{a}{l^2} \frf \left(\frac{r}{l^2}+\frac{M}{r^2}\right)$$
has a sign. To do this, one first checks that the function $\frac{1}{r^5}$ gives the expression a positive sign near infinity.  Next one continues $\frac{1}{r^5}$ at some fixed $r=R$ where $ r^4 \mathcal{O}(\frf \left(R\right)) \geq c\left(M,l\right)$ by solving the linear third order ODE $\mathcal{O}(\frf) = \mathcal{O}(\frf \left(R\right))$ with bounded coefficients on the bounded $r$-interval $\left[r_{min},R\right]$. (The latter fact guarantees that $\mathcal{F}$ itself stays bounded.)
On the other hand, looking back at (\ref{KXmain}), we see that the derivative term is controlled by the previous proposition and hence that we have the estimate
\begin{align} \label{zc}
 \int_{D\left(u,v\right)} \left(  \frac{1}{r^4} \phi^2 \right) \frac{\Omega^2}{2} r^2 \ du \ dv \  \leq B_{M,l} \cdot\mathbb{F}_{deg} \left(u,v\right) \, ,
\end{align}
i.e~control over the zeroth order term.

Step 2: To retrieve the missing derivative, we revisit (\ref{KXfm}). It is apparent that choosing $\mathcal{F}= \frac{1}{r^5}$ and $\mathfrak{f}=0$, we can dominate the mixed derivative term by the good-signed quadratic terms, while the zeroth order term is controlled by (\ref{zc}). Again, the boundary terms are obviously controlled by the energy. Hence
\begin{align}
 \int_{D\left(u,v\right)} \left(  \frac{1}{r^4} \phi^2 + \frac{1}{r^6} \frac{1}{\gamma^2} \left(\partial_u \phi\right)^2 +\frac{1}{r^6} \frac{1}{\kappa^2} \left(\partial_v \phi\right)^2    \right) \frac{\Omega^2}{2} r^2 \ du \ dv  \leq B_{M,l} \cdot\mathbb{F}_{deg} \left(u,v\right) \nonumber \, .
\end{align}

Step 3: Finally, we optimize the weights near infinity by choosing $\mathcal{F}= -\frac{1}{r}$ in (\ref{KXmain}). Note that the dominant term near infinity of the zeroth order term becomes
\begin{align}
\mathcal{O}(\frf) = \frac{a}{l^4} + \textrm{lower order terms}
\end{align}
in this case. On the other hand, the derivative term in (\ref{KXmain}) becomes \newline $\frac{1}{8r^2}\left(\frac{1}{\gamma} \phi_u - \frac{1}{\kappa} \phi_v\right)^2$. By Corollary \ref{heco}, applied with $n=1$ we can absorb the zeroth order term by the derivative term near infinity provided that $-a<\frac{9}{8}$. This finally yields the proposition, since an integrated decay estimate away from infinity was proved in Step 2.

Note that for any bounded $\mathcal{F}$ we have $\int_{D\left(u,v\right)} K_{error}^{X,\mathfrak{f}} \leq \epsilon \cdot {\mathbb{I}} \left[\phi\right] \left(D\left(u,v\right) \right)$, the $\epsilon$ coming from the pointwise bound on $\frac{\phi_u}{r_u}$ of Proposition \ref{uniformb}. Hence this error can always be absorbed by the main term.
\end{proof}
Up until now we proved
\begin{align} \label{zss}
{\mathbb{I}} \left[\phi\right] \left(D\left(u,v\right) \right) \leq B_{M,l} \cdot\mathbb{F}_{deg} \left(u,v\right)  \, .
\end{align}
The missing ingredient to reach Proposition \ref{intdecest} as stated is to go from the degenerate to the non-degenerate spacetime term, and to obtain the missing boundary term. This is achieved with the (future-directed, null) redshift vectorfield
\begin{align}
Y = \left(-r_u\right)^{-1} \partial_u \, .
\end{align}
From (\ref{Xid}) we obtain
\begin{align}
K^{Y,0} \left[\psi\right] = \frac{\left(\partial_u \psi\right)^2}{2r_u^2} \left(\frac{2\varpi}{r^2} + \frac{2r}{l^2} \right) + \frac{\partial_u \psi}{r_u} \frac{1}{r \kappa}\partial_v \psi + \frac{a}{l^2} \psi^2 \left[\frac{2}{r} \right] \, ,
\end{align}
\begin{align}
J^{Y,0} \left[\psi\right] \left(Y, \partial_u \right) = \frac{\left(\partial_u \psi\right)^2 }{r_u^2} \left(-r_u\right) \textrm{ \ \ \ \ \ , \ \ \ } 
J^{Y,0} \left[\psi\right] \left(Y, \partial_v \right) = \frac{2a \kappa}{2l^2} \psi^2 \, ,
\end{align}
which in turn leads to the estimate
\begin{align} \label{efg}
\int_{u_\mathcal{I}}^u \frac{\left(\partial_u \phi\right)^2 }{r_u^2} \left(-r_u\right) \left(\bar{u},v\right) d\bar{u} + \int_{D\left(u,v\right)}  \frac{r \left(\partial_u \phi\right)^2 }{r_u^2} \Omega^2 r^2 du dv \nonumber \\
\leq B_{M,l} \int_{D\left(u,v\right)} \left[ \frac{ \left(\partial_v \phi\right)^2 }{r^3 \kappa^2} + \frac{\phi^2}{r} \right] \Omega^2 r^2 d\bar{u} d\bar{v} +  B_{M,l} \int_{v_0}^v \frac{1}{l^2} \kappa \phi^2 r^2 \left(u, \bar{v}\right) d\bar{v} \nonumber \\ + \int_{u_\mathcal{I}}^u \frac{\left(\partial_u \phi\right)^2 }{r_u^2} \left(-r_u\right) \left(\bar{u},v_0\right) d\bar{u} \, .
\end{align}
The first term on the right hand side is controlled by the degenerate integrated decay $\mathbb{I}\left[\phi\right]\left(D\left(u,v\right)\right)$ that we already control by (\ref{zss}). The second term, a boundary term, can be converted into a spacetime-term, which is partly absorbed by the left and partly adds a term of the first type to the right hand side of (\ref{efg}). Indeed,
\begin{align}
\int_{v_0}^v \frac{1}{l^2} \kappa \phi^2 r^2 d\bar{v} = \int_{u_0}^u d\bar{u} \ \partial_u \int_{v_0}^{v^\star\left(u\right)} \frac{1}{l^2} \kappa \phi^2 r^2 d\bar{v} \, , 
\end{align}
where $v^\star \left(u\right)$ is the $v$-value where the ray of constant $u$ intersects either $\mathcal{I}$ or the constant $v$-ray (whatever happens first). The right hand side of the previous equation is equal to 
\begin{align} \label{eft}
= \frac{1}{l^2} \int_{D\left(u\right)} \left(-r \pi \frac{\left(\partial_u \phi\right)^2 }{r_u^2} \phi^2  -\frac{1}{2} \phi \frac{\partial_u \phi}{r_u}  - \frac{1}{2r} \phi^2 \right) \Omega^2 r^2 d\bar{u} d\bar{v} \, ,
\end{align}
observing that the boundary term on $\mathcal{I}$ vanishes in view of the decay of $\phi$ and that $v^\star \left(u\right)$ is constant in $u$ after the point where $u$ intersects the point where $v=const$ meets $\mathcal{I}$. A simple application of Cauchy's inequality then shows that (\ref{efg}) also holds without the $v$-boundary term.
Proposition \ref{intdecest} then follows since the last term in (\ref{efg}) is on data (requiring, however, the non-degenerate norm). 

\subsection{Proof of the estimate (\ref{ind1}) of Proposition \ref{intdec3}}
In view of the general identity
\begin{align} \label{gen2}
|\phi r^\frac{3}{2} \left(u,v\right)| \leq  B_{M,l} \sup_{D\left(u,v\right)} \|\phi\|_{H^1_{AdS}\left(u,v\right)} \, ,
\end{align}
we obtain from the estimate (\ref{cones}),
\begin{align}
\|\phi \|^2_{H^1_{AdS,deg}\left(u,v\right)} \leq B_{M,l,a} \left[ \|\phi \|^2_{H^1_{AdS,deg}\left(u,v_0\right)} + c^\frac{1}{3} \sup_{\left(\bar{u}, \bar{v}\right) \in D\left(u,v\right)} \|\phi \|^2_{H^1_{AdS} \left(\bar{u},\bar{v}\right)}  \right] . \nonumber
\end{align}
We can insert this estimate on the right hand side of the estimate of Proposition \ref{intdecest} and also add it to the resulting equation. This yields
\begin{align}
\|\phi \|^2_{H^1_{AdS}\left(u,v\right)} + \bar{\mathbb{I}} \left[\phi\right] \left(D\left(u,v\right) \right) 
\nonumber \\ 
\leq  B_{M,l,a} \left[ \|\phi \|^2_{H^1_{AdS,deg}\left(u,v_0\right)} + c^\frac{1}{3} \sup_{\left(\bar{u}, \bar{v}\right) \in D\left(u,v\right)} \|\phi \|^2_{H^1_{AdS} \left(\bar{u},\bar{v}\right)}  \right] \nonumber
\end{align}
Taking the $\sup$ in $D\left(u,v\right)$ and absorbing the $c^\frac{1}{3}$-term on the right yields (\ref{ind1}).
\section{Proof of Proposition \ref{uniformb2}: Improved and higher order bounds} \label{mp2}
For the estimates in this section, recall the initial data norm (\ref{idb}).
\subsection{Further consequences of the bootstrap assumptions}
In this section we will derive estimates for higher derivatives. These estimates are not sharp but sufficient to control the error in the commuted estimates later.
\begin{lemma} 
We have the pointwise estimate
\begin{align}  \label{secfree}
\Big| r^\frac{7}{2} \frac{\partial_u \frac{\phi_u}{r_u}}{r_u} \Big | \leq B_{M,l} \cdot \mathcal{N} \left[\phi\right] \left(v_0\right) 
\end{align}
\end{lemma}
\begin{proof}
We derive the following evolution equation:
\begin{align}
\partial_v \left(\frac{\partial_u \frac{\phi_u}{r_u}}{r_u} \right) = \left[ -4 \kappa \left(\frac{\varpi}{r^2} + \frac{r}{l^2} -\frac{4\pi a \phi^2}{l^2} + \frac{1-\mu}{4r} \right)  \right] \left(\frac{\partial_u \frac{\phi_u}{r_u}}{r_u} \right) \nonumber \\ 
 +\frac{2}{r^2} \phi_v
- 8\pi r \frac{\kappa a}{l^2} \phi \left( \frac{\phi_u}{r_u} \right)^2 - \frac{2 \kappa a \phi }{l^2 r}
+ \frac{\phi_u}{r_u} \left( \frac{2 \lambda}{r^2} - \frac{1}{r r_u} \partial_u \left( r \frac{r_{uv}}{r_u} \right) + \frac{2a\kappa}{l^2} \right) 
\end{align} 
Note the exponential decay factor (redshift) in the square bracket in the first line. Integrating and estimating the errors as in Lemma \ref{zetahozbound} yields the result. See also section 7 of our \cite{gs:lwp}, where the same computation is carried out.
\end{proof}
Interestingly enough, we were able to derive this pointwise bound for a second derivative of $u$ without the need for second $v$ derivatives of any quantity. 

\subsection{The wave equation for $T\left[\phi\right]= \frac{1}{4\kappa}\partial_v \phi + \frac{1}{4\gamma} \partial_u \phi$}
We turn to the commutation of the wave equation by the vectorfield $T$. The following Lemma may be found in the appendix of \cite{Mihalisnotes}:
\begin{lemma} \label{coml}
Let $\psi$ be a solution of the equation $\Box_g \psi = 0$ and $X$ be a vectorfield. Then
\begin{align}
\Box_g \left(X\psi\right) =  \mathfrak{q}\left[X\psi\right] \, \textrm{ \ \ \ with \ }
\end{align}
\begin{align}
\mathfrak{q}\left[X\psi\right] = -2 {}^{(X)}\pi^{\alpha \beta} \nabla_a \nabla_b \psi - 2 \left[2 \nabla^\alpha {}^{(X)}\pi_{\alpha \mu} - \nabla_\mu \left( tr {}^{(X)}\pi \right) \right] \nabla^\mu \psi \, .
\end{align}
\end{lemma}
$\phantom{X}$ \newline
From the computations
\begin{align}
-4 g^{\alpha \beta} \left(\nabla_\alpha \pi_{\beta \gamma}\right) \nabla^\gamma \psi = -4 \left(\frac{2}{\Omega^2}\right)^2 \left[\partial_u \psi \ \partial_u \left(2\pi r \phi_v^2\right) - \partial_v \psi \ \partial_v \left(2\pi r \phi_u^2\right) \right] \nonumber \\
+ \frac{16}{\gamma r \Omega^2} \partial_v \psi \left(\pi r \phi_u^2\right) + \frac{16}{\kappa r \Omega^2} \partial_u \psi \left( \pi r \phi_v^2\right) \nonumber \\
=  \frac{8}{\gamma r \Omega^2} \partial_v \psi \left(\pi r \phi_u^2\right) + \frac{8}{\kappa r \Omega^2} \partial_u \psi \left( \pi r \phi_v^2\right) + \pi \frac{16}{\Omega^2 \gamma} \psi_u \phi_u \phi_v + \pi \frac{16}{\Omega^2 \kappa} \psi_v \phi_u \phi_v \nonumber \\
 + 32 \frac{1}{\Omega^2} r \frac{a}{l^2} \left(\partial_u \psi\right) \phi \phi_v -  32 \frac{1}{\Omega^2} r \frac{a}{l^2} \left(\partial_v \psi\right) \phi \phi_u \nonumber 
\end{align}
and
\begin{align}
-2 {}^{(X)}\pi^{\alpha \beta} \nabla_a \nabla_b \psi = -2 \pi^{\alpha \beta} \left(\partial_{\alpha} \partial_{\beta} \psi - \Gamma^{\delta}_{\alpha \beta} \partial_{\delta} \psi \right) \nonumber \\
= -16 \frac{\pi r}{\Omega^2}  \left[ \left(\partial_v \phi\right)^2 \partial_u \left(\frac{\partial_u \psi}{\Omega^2}\right) -  \left(\partial_u \phi\right)^2 \partial_v \left(\frac{\partial_v \psi}{\Omega^2}\right)\right] \nonumber \\
= - \frac{\pi r \phi_v^2}{\kappa^2} \left[\frac{\partial_u \left(\frac{\partial_u \psi}{r_u}\right)}{r_u}\right] + \frac{4\pi^2 r^2 \phi_v^2}{\kappa^2} \left(\frac{\phi_u}{r_u}\right)^2 \frac{\psi_u}{r_u} \nonumber \\
+ 16 \frac{\pi r}{\Omega^2} \left(\partial_u \phi\right)^2 \left[-\frac{1}{r_u} \partial_v \left(T\left[\psi\right] \right) + \frac{r_{uv}}{r_u} \frac{1}{r_u} \frac{1}{\kappa} \partial_v \psi - \frac{4}{\Omega^2} r \pi \phi_v^2 \frac{\phi_u}{r_u} + \frac{1}{4\gamma r_u} \phi_{uv} \right] \nonumber \, ,
\end{align}
we find using the Lemma applied with $\psi = \phi$,
\begin{align} \label{met}
\mathfrak{q} \left[T\phi\right] =  4\pi r  \left(\frac{\phi_u}{r_u}\right)^2 \frac{1}{\kappa} \partial_v \left(T\left[\phi\right] \right) - \frac{\pi r \phi_v^2}{\kappa^2} \left[\frac{\partial_u \left(\frac{\partial_u \phi}{r_u}\right)}{r_u}\right] \nonumber \\
+ \left(6\pi + 16\pi r \frac{r_{uv}}{\Omega^2} +\pi \left(1-\mu\right) \right) \frac{\phi_v}{\kappa} \left(\frac{\phi_u}{r_u}\right)^2 - 6\pi \frac{\phi_v^2}{\kappa^2} \left(\frac{\phi_u}{r_u}\right) \nonumber \\ +  \pi \left(\frac{\phi_u}{r_u}\right)^3 \left(1-\mu\right)^2 
+ \frac{2\pi a r}{l^2} \left(1-\mu\right) \phi  \left(\frac{\phi_u}{r_u}\right)^2 - 4 r^2 \pi^2 \left( \frac{\phi_v}{\kappa}\right)^2 \left(\frac{\phi_u}{r_u}\right)^3 \, .
\end{align}
\subsection{Estimates for $T\left[\phi\right]$}
For any point $\left(u,v\right)$ in $\mathcal{R}_\mathcal{H}$, we define the higher order energy
\begin{align}
E \left[T\phi\right] \left(u,v\right) = \int_{u_{\mathcal{I}}}^u \Big[ 2 \pi r^2 \frac{1}{\gamma} \left(\partial_u T\left[\phi\right]\right)^2- 4a \pi \frac{r^2}{l^2}r_u \left(T\left[\phi\right]\right)^2 \Big] \left(\bar{u},v\right) d\bar{u} \nonumber \\
+ \int_{v_0}^v \Big[ 2 \pi r^2 \frac{1}{\kappa} \left(\partial_v T\left[\phi\right]\right)^2 + 4a \pi \frac{r^2}{l^2}r_v \left(T\left[\phi\right]\right)^2 \Big] \left(u,\bar{v}\right) d\bar{v}  \, .
\end{align}
\begin{lemma}
The energy $E \left[T\phi\right] \left(u,v\right)$ is almost conserved in that
\begin{align} \label{amozi}
E \left[T\phi\right] \left(u,v\right) = E \left[T\phi\right] \left(u,v_0\right) + \int_{D\left(u,v\right)} \mathcal{Q} \ \  \Omega^2 r^2 du dv \nonumber \\
\textrm{where \ \ \ }\mathcal{Q} = \frac{1}{\Omega^2 r^2} \Bigg( - 32 r^3 \pi^2 \frac{1}{\Omega^2}  \left(\partial_v \phi\right)^2 \left(\partial_u T\left[\phi\right]\right)^2  \nonumber \\+ 32 r^3 \pi^2 \frac{1}{\Omega^2}  \left(\partial_u \phi\right)^2 \left(\partial_v T\left[\phi\right]\right)^2 + \pi r^2 \Omega^2 \cdot TT\left[\phi\right] \cdot \mathfrak{q}\left[T\phi\right] \Bigg)  
\end{align}
holds.
\end{lemma}
\begin{proof}
The quantity $T\left[\phi\right]$ satisfies the wave equation $\Box_g T\left[\phi\right] = \mathfrak{q} \left[T\phi\right]$. Integrating the energy identity (\ref{Xid}) for the energy momentum tensor associated with $T\left[\phi\right]$ with the vectorfield $T = \frac{1}{4\kappa} \partial_v + \frac{1}{4\gamma} \partial_u$ yields the above identity using that the energy-flux through the boundary at infinity is zero by the local well-posedness result  \cite{gs:lwp}.
\end{proof}
\begin{lemma} \label{mltp}
For any $\left(u,v\right) \in \mathcal{R}_{\mathcal{H}}$ we have the estimate
\begin{align} \label{hiotp}
 \|T\phi \|_{H^1_{AdS}\left(u,v\right)} \leq B_{M,l} \cdot \mathcal{N} \left[\phi\right] \left(v_0\right) + B_{M,l} \left(\int_{D\left(u,v\right)} |\mathcal{Q}| \ \  \Omega^2 r^2 du dv \right)^\frac{1}{2}
\end{align}
\end{lemma}
\begin{proof}
We can prove this estimate in the same way that we proved it for $\phi$ itself in the context of the bootstrap of section \ref{mp1} (in which case there was no commutation error $\mathcal{Q}$, of course). In fact, a bootstrap is no-longer necessary, since the important Hardy inequalities have already been established in $\mathcal{R}_{\mathcal{H}}$. 

Integrating $T\left[\phi\right]$ from infinity one derives the pointwise bound
\begin{align}
|T\phi| \leq B_{M,l} \cdot r^{-\frac{3}{2}} \cdot \|T\phi \|_{H^1_{AdS,deg}\left(u,v\right)}  \textrm{ \ \ \ \ in $\mathcal{R}_{\mathcal{H}} \cap \{r\geq r_X$\}}
\end{align} 
using Cauchy-Schwarz as in Lemma \ref{phibound}. Repeating the redshift estimate of Lemma \ref{zetahozbound}, now with $\partial_v \left(\frac{\partial_u T\left[\phi\right]}{r_u}\right)$, we derive for any $\left(u,v\right) \in \mathcal{R}_{\mathcal{H}} \cap \{r\leq r_X\}$
\begin{align} \label{rsTu}
\Big| \frac{\partial_u \left(T\phi\right)}{\nu} \Big| + |T\phi| \leq B_{M,l} \left[ \sup_{\left(\bar{u},\bar{v}\right)\in D\left(u,v\right)} \|T\phi\|_{H^1_{AdS,deg} \left(\bar{u},\bar{v}\right)} + \mathcal{N} \left[\phi\right] \left(v_0\right) \right] \, .
\end{align}
Indeed, the only difference to Lemma \ref{zetahozbound} is that there is now an error-term $r \kappa \mathfrak{q}\left[T\psi\right]$ on the right hand side of the equation due to the commutation term in the wave equation. For this error-term we note that (using definition (\ref{rsweight}))
\begin{align}
\left| \int_{v_0}^v d\bar{v} \left[ \exp \left(-\int_{\bar{v}}^v \rho \left(u, \hat{v} \right) d\hat{v}\right) r \kappa \mathfrak{q}\left[T\psi\right] \right]\right| \nonumber \\
\leq \epsilon \left(  \sup_{\left(\bar{u},\bar{v}\right)\in D\left(u,v\right)} \|T\phi\|_{H^1_{AdS,deg} \left(\bar{u},\bar{v}\right)}  + B_{M,l} \cdot  \mathcal{N} \left[\phi\right] \left(v_0\right)  \right) \, ,
\end{align}
which follows by inspecting the terms in $\mathfrak{q}\left[T\phi\right]$ individually:
\begin{itemize}
\item the first term is estimated as in Lemma \ref{zetahozbound}, cf.~the estimate (\ref{fioc}). The same holds for the first term in the second line. Note the smallness factor arising from the pointwise bound on $\frac{\phi_u}{r_u}$, Proposition \ref{uniformb}.
\item terms which have $\phi_v^2$ can be estimated using the pointwise bound on $\frac{\phi_u}{r_u}$ and (\ref{secfree}) and the $H^1_{AdS,deg}$-norm for $\phi$. Note that $\|\phi\|^2_{H^1_{AdS,deg} \left(u,v\right)} \geq B_{M,l} \cdot \epsilon \cdot \|\phi\|^2_{H^1_{AdS,deg} \left(u,v\right)}$ from Proposition \ref{uniformb}.
\item the terms which have $\left(1-\mu\right)$ have $\lambda^2$ and are easily integrated using the pointwise bound on $\frac{\phi_u}{r_u}$

\end{itemize}
With the pointwise bound on $|T\phi|$ we repeat the conservation of energy argument (\ref{cones}), now modified to the almost conservation (\ref{amozi}):
For any $\left(u,v\right) \in \mathcal{R}_{\mathcal{H}}$,
\begin{align}
 \|T\phi\|^2_{H^1_{AdS,deg}\left(u_{\mathcal{H}},v_0\right)} \geq \int_{u_0}^u  \left(2\pi \frac{\left(\partial_u \left[T\phi\right]\right)^2}{\gamma} - \frac{4\pi a}{l^2}\left(T\phi\right)^2 r_u \right)  r^2 \left(\bar{u} ,v_0\right) d\bar{u} \nonumber \\
= \int_{u_\mathcal{I}}^u  \left(\mathbf{1}_{\{r\geq r_Y\}} +\mathbf{1}_{\{r\leq r_Y\}} \right) \left(2\pi \frac{\left(\partial_u\left[T\phi\right]\right)^2}{\gamma} - \frac{4\pi a}{l^2}\left(T\phi\right)^2 r_u \right)  r^2 \left(\bar{u} ,v\right) d\bar{u} +
\nonumber \\ 
 \int_{v_0}^v   \left(\mathbf{1}_{r\geq r_Y} +\mathbf{1}_{r\leq r_Y} \right) \left(2\pi \frac{\left(\partial_v\left[T\phi\right]\right)^2}{\kappa} + \frac{4\pi a}{l^2}\left(T\phi\right)^2 r_v \right)  r^2 \left(u ,\bar{v}\right) d\bar{v} 
+ \int_{D\left(u,v\right)} \mathcal{Q} \nonumber
\end{align}
Using the Hardy inequalities of Lemma \ref{lem:vhardy} (which clearly hold for $\phi$ replaced by any $\psi$ satisfying the same boundary conditions at infinity, hence in particular for $T\phi$), we continue to estimate
\begin{align}
 \|T\phi\|^2_{H^1_{AdS,deg}\left(u_{\mathcal{H}},v_0\right)} \geq \int_{D\left(u,v\right)} \mathcal{Q} + \frac{9}{8}\frac{4\pi}{l^2} \int_{u_\mathcal{I}}^u \mathbf{1}_{r\leq r_Y} |T\phi|^2 r_u r^2 \left(\bar{u} ,v\right) d\bar{u} \nonumber \\ - \frac{9}{8}\frac{4\pi}{l^2} \int_{v_0}^v  \mathbf{1}_{r\leq r_Y} |T\phi|^2 r_v r^2 \left(u ,\bar{v}\right) d\bar{v} +  \frac{1}{2}\left(a+\frac{9}{8}\right)  \|T\phi\|^2_{H^1_{AdS,deg}\left(u,v\right)} \nonumber \\
 \geq \left(a+\frac{9}{8}\right)  \|T\phi\|^2_{H^1_{AdS,deg}\left(u,v\right)} + \int_{D\left(u,v\right)} \mathcal{Q}  \nonumber \\
 - B_{M,l} c^\frac{1}{3}  \left( \sup_{\mathcal{R}_{\mathcal{H}}\cap \{\bar{v}\leq v\} \cap  \{\bar{u}\leq u\}} \|T\phi\|_{H^1_{AdS,deg} \left(\bar{u},\bar{v}\right)}  + B_{M,l} \cdot  \mathcal{N} \left[\phi\right] \left(v_0\right)  \right) \nonumber \, ,
\end{align}
where we inserted the pointwise estimate (\ref{rsTu}) in the last step and exploited the $c^\frac{1}{3}$ smallness of the $r$-difference is the region $r\leq r_Y$. Taking the $\sup$ over all ${\left(\bar{u},\bar{v}\right)\in D\left(u,v\right)}$ of this estimate we arrive at the estimate of Lemma \ref{mltp} for the $H^1_{AdS,deg}$-norm on the left hand side. However, in view of (\ref{rsTu}) and the remarks of section \ref{conclusio}, the estimate also holds for the $H^1_{AdS}$-norm.
 \end{proof}
\begin{corollary} \label{rsTp}
For any $\left(u,v\right) \in \mathcal{R}_{\mathcal{H}}$ we have the estimate
\begin{align} \label{sep}
 \Big| r^\frac{3}{2} \frac{r \partial_u \left(T\phi\right)}{r_u} \Big|  \leq B_{M,l} \cdot \mathcal{N} \left[\phi\right] \left(v_0\right) + B_{M,l} \left(\int_{D\left(u,v\right)} |\mathcal{Q}| \ \  \Omega^2 r^2 du dv \right)^\frac{1}{2}  \, .
\end{align}
\end{corollary}
\begin{proof}
This is immediate in $r\leq r_X$ from (\ref{rsTu}) and the Lemma \ref{mltp}. For $r\geq r_X$ one repeats the proof of Lemma \ref{erb}.
\end{proof}

Since $T\phi$ satisfies the wave equation with an inhomogeneous error-term on the right hand side, we can prove the same integrated decay estimate for $T\phi$ that we proved for $\phi$, corrected only by the error-term arising from commutation: 
\begin{lemma} \label{eef}
For any $\left(u,v\right) \in \mathcal{R}_{\mathcal{H}}$ we have the estimate
\begin{align} \label{amoc2}
 \|T\left[\phi\right] \|^2_{H^1_{AdS}\left(u,v\right)} + \overline{\mathbb{I}} \left[T\phi\right] \left(D\left(u,v\right) \right)  \leq & \ B_{M,l} \cdot \mathcal{N}^2 \left[\phi\right] \left(v_0\right) 
\nonumber \\  &+B_{M,l} \int_{D\left(u,v\right)} \mathcal{P}  \  \Omega^2 r^2 du dv 
\end{align}
with
\begin{align} \label{ame}
\mathcal{P}  = \left( \Big|\frac{r}{r_u} \partial_u \left[T\phi\right] \Big| + \frac{1}{r \kappa} \Big|\partial_v \left[T\phi\right] \Big| + \frac{1}{r^2} \Big| T\left[\phi\right] \Big| \right) \Big|\mathfrak{q} \left[T\phi\right] \Big| 
 + |\mathcal{Q}| \, .
\end{align} 
\end{lemma}
\begin{proof}
We are proving an estimate for the same wave equation as in Proposition \ref{intdecest}, except that there is the inhomogeneity $q\left[T\phi\right]$ on the right hand side. Looking at formula (\ref{mfvf}), we see that this inhomogeneous term enters the vectorfield estimates as the spacetime error-term
\begin{align}
\int_{D\left(u,v\right)} \left(X\left[T\phi\right] + \mathfrak{f}\ T\phi\right) q\left[T\phi\right] \, .
\end{align}
Checking carefully which pairs of multipliers $\left(X,\mathfrak{f}\right)$ were used to derive the integrated decay estimate, we see that the first term in (\ref{ame}) accounts for these terms. The boundary terms in the integrated decay estimate are again controlled by the $\|T\phi\|_{H^1_{AdS}\left(u,v\right)}$-norm. Inserting the estimate (\ref{hiotp}) for these will produce the last term in (\ref{ame}).
\end{proof}
\begin{lemma} \label{eee}
Let $|\frac{\phi_v}{\kappa}|<1$ hold in $\tilde{D}\left(u,v\right)$.
Then 
\begin{align}
\int_{\tilde{D}\left(u,v\right)} \mathcal{P}  \  \Omega^2 r^2 du dv  \leq B_{M,l} \cdot \epsilon \cdot \Big( \overline{\mathbb{I}} \left[T\phi\right] \left(\tilde{D}\left(u,v\right) \right) +  \overline{\mathbb{I}} \left[\phi\right] \left(\tilde{D}\left(u,v\right) \right) \Big) \, ,
\end{align}
with the $\epsilon$-factor arising from the smallness of the $\mathcal{N} \left[\phi\right] \left(v_0\right)$ norm.
\end{lemma}
\begin{proof}
Inspecting the terms in (\ref{met}), this is an easy application of Cauchy's inequality after using the pointwise bounds for $\phi_v$ from the assumption and the smallness bounds we already established for $\frac{\phi_u}{r_u}$ and (\ref{secfree}).
\end{proof}

We can finally derive the second estimate  (\ref{ind2}) of Proposition \ref{intdec3} by a bootstrap on the size of $|\frac{\phi_v}{\kappa}|<1$. Recall the region $\widehat{\mathcal{B}} \left(\tilde{u}\right)$ from (\ref{bwh}).
Let
\begin{align}
u_{max} = \sup_u \Big( \textrm{ $\Big|\frac{\phi_v}{\kappa}\Big|<1$ \ holds in \ $\widehat{\mathcal{B}}\left(u\right)$} \, \Big) \, .
\end{align}
and $\mathcal{B} = \widehat{\mathcal{B}} \left(u_{\max} \right)$. By the local well-posedness, $\mathcal{B}$ is non-empty and by continuity of the pointwise norm, the region $\mathcal{B}$ is open. We show that $\mathcal{B}$ is also closed, which implies $\mathcal{B} = \mathcal{R}_{\mathcal{H}}$.
Combining Lemma \ref{eef} and Lemma \ref{eee} with Proposition \ref{intdecest}, we obtain
that for $\left(u,v\right) \in \mathcal{B}$ we have
\begin{align} \label{axl}
 \Big| r^\frac{5}{2} \frac{\partial_u \left(T\phi\right)}{r_u} \Big| +  \|T\left[\phi\right] \|^2_{H^1_{AdS}\left(u,v\right)} + \overline{\mathbb{I}} \left[T\phi\right] \left(D\left(u,v\right) \right)  \leq B_{M,l} \cdot \mathcal{N}^2 \left[\phi\right] \left(v_0\right).
\end{align}
Integrating $T\phi = 0 + \int du \ \partial_u \left(T\phi\right) du$ from infinity (where $T\phi$ vanishes) yields in view of the previous estimate,
\begin{align} \label{Tphidecay}
|T\left(\phi\right)| \leq  B_{M,l} \cdot \mathcal{N} \left[\phi\right] \left(v_0\right)  \cdot  r^{-\frac{3}{2}} \, .
\end{align}
Finally, from the relation $\phi_v = \kappa T\left(\phi\right) + \kappa \left(1-\mu\right) \frac{\phi_u}{r_u}$ and the pointwise bounds already established for the right hand side we obtain in particular $\frac{\phi_v}{\kappa}<\frac{1}{2}$ in $\mathcal{B}$. The bootstrap closes, hence (\ref{axl}) holds in all of $\mathcal{R}_{\mathcal{H}}$, which implies both the estimate  (\ref{ind2}) of Proposition \ref{intdec3} and (\ref{sep1}) of Proposition \ref{uniformb2}.

\subsection{Improved $r$-weighted bounds for $\phi$ and first derivatives}
It remains to establish (\ref{sep2}) of Proposition \ref{uniformb2}:
\begin{lemma} \label{improvelem}
We have, in $\mathcal{R}_\mathcal{H} \cap \{ r\geq r_X \}$ the bounds
\begin{align}
 | \phi \left(u,v\right) | \leq  C_\delta \cdot B_{M,l} \cdot \mathcal{N}\left[\phi\right] \left(v_0\right) \cdot  r^{max\left(2p-3,-\frac{5}{2}+\delta\right)} 
\end{align}
\begin{align}
\Big|r^2 \frac{\phi_u}{r_u} \left(u,v\right)\Big| +  |\phi_v \left(u,v\right)| \leq C_\delta \cdot  B_{M,l} \cdot \mathcal{N}\left[\phi\right] \left(v_0\right)  \cdot  r^{max\left(2p-4,-\frac{3}{2}+\delta \right)}
\end{align}
for any $\delta>0$.
\end{lemma}

\begin{proof}
On the initial data and on $r=r_X$ these bounds hold by assumption and Proposition \ref{uniformb} respectively. Let now
\begin{align}
p = \frac{3}{4} - \sqrt{\frac{9}{16} - \frac{-a}{2}} \, .
\end{align}
One derives the following evolution equation for 
\begin{align}
A = r^n \frac{\zeta}{\nu} + 2p r^n \phi \, :
\end{align}
\begin{align} \label{icvd}
\partial_v A &= A\left[ \frac{\lambda}{r} \left(n+2p-1\right) - \rho \right] + f  \nonumber \\
f &= \left(2p-1\right) r^n \kappa T \left(\phi\right) + 2\phi \kappa r^{n+1} p \left[ \frac{1}{r^2} \left(1-2p\right) + \frac{4\varpi p}{r^3} - 8\pi \frac{a}{l^2} \phi^2 \right]
\end{align}
where we recall the redshift weight $\rho$ from (\ref{rsweight}). This computation exploits an important cancellation: The zeroth order term in $f$ decays better (in $r$) than naively expected while we have already shown improved decay for $T\left(\phi\right)$ by our commutation argument. Note in this context that the conformally coupled case, $p=\frac{1}{2}$ is special.\footnote{In particular, commutation by $T$ is not necessary to obtain the improved estimates of the Proposition.}
Noting that
\begin{align}
\frac{\lambda}{r} \left(n+2p-1\right) - \rho = \frac{\kappa}{r} \left( \frac{n-3 + 2p}{l^2} r^2 + \left(n+2p-1\right) + \textrm{terms decaying in $r$}\right) \nonumber
\end{align}
we choose $n=\min \left(3-2p, \frac{5}{2}-\delta\right)$. Note that for $n=3-2p$ we have (using that $\kappa \leq 8d^{-\frac{1}{3}}\frac{l^2 \lambda}{r^2}$ in $r\geq r_X$)
\begin{align}
\int_{v_0}^v \mathbf{1}_{r\geq r_X} \left(\frac{\lambda}{r} \left(n+2p-1\right) - \rho\right) \left(u,\bar{v}\right) d\bar{v} \leq B_{M,l}
\end{align}
uniformly, while for $n=\frac{5}{2}-\delta$ we obtain an exponential decay factor in (\ref{icvd}). Either way, integrating (\ref{icvd}), one easily obtains the following estimate for $A$  in all of $r\geq r_X$:
\begin{align}
|A\left(u,v\right)| \leq B_{M,l} \cdot  \int_{v_0}^v \mathbf{1}_{r\geq r_X} |f| \left(u,\bar{v}\right) d\bar{v} \, .
\end{align}
To estimate this, we exploit the pointwise bounds available for both $\phi$ and $T\phi$. For instance, from (\ref{Tphidecay}),
\begin{align}
\int_{v_0}^v \mathbf{1}_{r\geq r_X} r^n|\kappa T\phi| \left(u,\bar{v}\right) d\bar{v}  \leq B_{M,l} \cdot \mathcal{N}\left[\phi\right] \left(v_0\right) \int_{v_0}^v \mathbf{1}_{r\geq r_X} r^{n-\frac{3}{2}} \frac{r_v}{r^2} \left(u,\bar{v}\right) d\bar{v} \nonumber \\
\leq B_{M,l} \cdot \mathcal{N}\left[\phi\right] \left(v_0\right)C_\delta \nonumber \, .
\end{align}
The $\phi$-term can be estimated in the same way. What we have shown so far is
\begin{align} \label{imlat}
\Big|r^{\min\left(3-2p,\frac{5}{2}-\delta\right)} \left( \frac{\zeta}{\nu} + 2p \phi \right) \Big| \leq B_{M,l} \cdot \mathcal{N}\left[\phi\right] \left(v_0\right)C_\delta   \, ,
\end{align}
which in view of $2p<\frac{3}{2}$ is an improvement over previous estimates of the two summands alone.

With the improved decay for the quantity above one can re-estimate $\phi$ from infinity by integrating $r^{2p} \phi$. The latter quantity vanishes at infinity, since $2p< \frac{3}{2}$, which is the decay we already established for $\phi$.
\begin{align}
r^{2p} \phi \left(u,v\right) = \int_{u_{\mathcal{I}\left(v\right)}}^u  \frac{\partial_u \left(r^{2p}\phi\right)}{\nu} \nu  \left(\bar{u},v\right) d\bar{u} \nonumber \\  \leq B_{M,l} \cdot \mathcal{N}\left[\phi\right] \left(v_0\right) \int_{u_{\mathcal{I}\left(v\right)}}^u r^{\max\left(4p-4, 2p-\frac{7}{2}+\delta \right)} \left(-\nu\right) \left(\bar{u},v\right) d\bar{u} \nonumber \\
\leq C_\delta \cdot B_{M,l} \cdot \mathcal{N}\left[\phi\right]\left(v_0\right) r^{\max\left(4p-3, 2p-\frac{5}{2}+\delta \right)} \, .
\end{align}
and hence
\begin{align} 
| \phi| \left(u,v\right) \leq C_\delta \cdot B_{M,l} \cdot \mathcal{N}\left[\phi\right]\left(v_0\right) \cdot  r^{max\left(2p-3,-\frac{5}{2}+\delta\right)} \, ,
\end{align}
which is the first estimate of the Lemma. The second immediately follows by combining it with the estimate (\ref{imlat}). 
For the bound on $\phi_v$ note that $\phi_v = \frac{\Omega^2}{-r_u} T\left(\phi\right) - r_v \frac{\phi_u}{r_u}$ and use the previous bounds.
\end{proof}
\begin{remark} \label{moremore}
For $3-2p>\frac{5}{2}$, one can actually improve the decay further by another commutation with $T$ (which will improve the pointwise decay for $T\phi$ to what we have just shown for $\phi$) to establish the heuristically expected $r^{3-2p}$ decay for $\phi$. Since the gain is not needed, we do not concern ourselves with optimizing the result in that direction.
\end{remark}
\section{Exponential decay in $v$} \label{expvdec}
Given the estimate (\ref{ind1}), we can show that the solution decays exponentially in $v$. Define the flux
\begin{align}
\mathbb{F} \left(v\right) = \int_{u_{\mathcal{I}}}^{u_{\mathcal{H}}} \left[\frac{\left(\partial_u \phi\right)^2}{-r_u} r^4 + \phi^2 \ r^2 \left(-r_u\right) \right] \left(\bar{u},v\right) d\bar{u} 
\end{align}
and the region $D^\star \left(v_1,v_2\right) = D\left(u_{\mathcal{H}},v_2\right) \cap \{ v \geq v_1 \}$.
Note that 
\begin{align}
\bar{\mathbb{I}} \left[\phi\right] \left(D^\star\left(v_1,v_2\right)\right) \geq b_{M,l} \int_{v_1}^{v_2} d\bar{v} \ \mathbb{F}\left(\bar{v}\right)
\end{align}
where $b_{M,l}$ is a small positive constant depending only on $M$ and $l$.
Applying the estimate (\ref{ind1}) in the region $D^\star\left(v_1,v_2\right)$ (i.e.~not starting from $v=v_0$ but from $v=v_1$) yields
\begin{align}
\mathbb{F} \left(v_2\right) + b_{M,l} \int_{v_1}^{v_2} d\bar{v} \  \mathbb{F}\left(\bar{v}\right) \leq B_{M,l,a} \cdot \mathbb{F} \left(v_1\right) 
\end{align}
which implies exponential decay for $\mathbb{F} \left(v\right)$. From the estimate
\begin{align}
|\phi \left(u,v\right) | \leq  B_{M,l,a} \cdot \frac{1}{r^\frac{3}{2}} \cdot \sqrt{\mathbb{F} \left(v\right)} \, ,
\end{align}
which easily follows from Cauchy Schwarz, we conclude that $|\phi|$ decays exponentially in $v$, which is the last claim of Theorem \ref{maintheorem}
\appendix

\section{Absence of stationary solutions in the linear case} \label{noha}
We present here an easy argument to establish the non-existence of stationary solutions for the wave equation on Schwarzschild-AdS backgrounds satisfying the boundary conditions of \cite{HolzegelAdS, Holzegelwp}. 

Assume that there was a stationary solution $\psi$ of (\ref{wap}) on a fixed AdS-Schwarzschild background,
\begin{align}
g = - F\left(r\right)dt^2 + F\left(r\right)^{-1} dr^2 +r^2 \left(d\theta^2 + \sin^2 d\varphi^2\right) \ \  , \ \ F\left(r\right) = 1- \frac{2M}{r} + \frac{r^2}{l^2} \, . \nonumber
\end{align}
In view of $\partial_t \psi =0$, $\psi$ must satisfy
\begin{align}
\frac{1}{r^2} \partial_m \left(r^2 g^{mn} \partial_n \psi\right) + \frac{2a}{l^2} \psi = 0  \textrm{ \ \ \ \ \ with $m,n = \{r, \theta, \varphi\}$} \, .
\end{align}
Multiplying this equation by $r^2 \psi$ and integrating over a constant $t$-slice with $dr d\theta d\varphi$ we obtain after integrating by parts,
\begin{align}
\int dr \ d\theta \ d\varphi \  r^2 \left[F\left(r\right) \left(\partial_r \psi\right)^2 + r^2 g^{AB} \partial_A \psi \partial_B \psi - \frac{2a}{l^2} \psi^2 \right] = 0 \, .
\end{align}
Note that the boundary-terms vanish both at infinity (in view of the boundary conditions of \cite{HolzegelAdS}) and at the horizon (since $g^{rr} = F\left(r\right)=0$ there) in this computation.
By the Hardy inequalities proven in \cite{HolzegelAdS} this implies that $\psi = 0$, as the zeroth order term can be absorbed by the derivative term for $-a>\frac{9}{8}$. Hence there are no non-trivial stationary solutions for the wave equation on Schwarzschild satisfying the boundary conditions. 
\bibliographystyle{hacm}
\bibliography{thesisrefs}
\end{document}